\documentclass[12 pt]{amsart}

\usepackage{amssymb, amsmath}
\usepackage[margin=3.5 cm]{geometry}
\usepackage{upgreek}
\usepackage{graphicx}
\usepackage{color}
\usepackage{comment} 
\usepackage{subfigure}
\newtheorem{theorem}{Theorem}[section]
\newtheorem{prop}{Proposition}[section]
\newtheorem{claim}{Claim}[section]

\newtheorem{coro}{Corollary}[section]
\newtheorem{lemma}[theorem]{Lemma}
\newtheorem{conj}{Conjecture}[section]
\newtheorem{obs}{Observation}[section]

\theoremstyle{definition}
\newtheorem{definition}[theorem]{Definition}

\theoremstyle{remark}
\newtheorem{remark}[theorem]{Remark}

\numberwithin{equation}{section}

\newcommand{\abs}[1]{\lvert#1\rvert}
\newcommand{\norm}[1]{\lVert#1\rVert}
\DeclareMathOperator{\re}{Re}
\DeclareMathOperator{\im}{Im}
\newcommand{\ud}{\mathrm{d}}

\begin{document}

\title[Lyapunov exponent of Extended Harper's Equation]{Analytic quasi-perodic cocycles with singularities and the Lyapunov Exponent of Extended Harper's Model}

\author{S. Jitomirskaya and C. A. Marx}
\address{Department of Mathematics, University of California, Irvine CA, 92717}
%\email{cmarx@math.uci.edu}
\thanks{The work was supported by NSF Grant DMS - 0601081 and BSF, grant 2006483. .}

%    General info
%\subjclass[2000]{Primary 54C40, 14E20; Secondary 46E25, 20C20}

%\date{January 1, 2001 and, in revised form, June 22, 2001.}


%\keywords{Differential geometry, algebraic geometry}

\begin{abstract}
We show how to extend (and with what limitations) Avila's global theory of analytic SL(2,C) cocycles to families of cocycles with singularities. This allows us to develop a strategy to determine the Lyapunov exponent for extended Harper's model, for all values of parameters and all irrational frequencies. In particular, this includes the self-dual regime for which even heuristic results did not previously exist in physics literature. The extension of Avila's global theory is also shown to imply continuous behavior of the LE on the space of analytic $M_2(\mathbb{C})$-cocycles. This includes rational approximation of the frequency, which so far has not been available.
\end{abstract}

\maketitle

\section{Introduction} \label{sec_intro}
For an irrational $\beta$, consider the quasi-periodic Jacobi operators on $l^2(\mathbb{Z})$
\begin{eqnarray} \label{eq_hamiltonian}
& (H_{\theta;\beta} \psi)_k := v(\theta + \beta k) \psi_{k} + c(\theta + \beta k) \psi_{k+1} + \overline{c}(\theta + \beta (k-1)) \psi_{k-1} ~\mbox{.}
\end{eqnarray}
indexed by $\theta \in [0,1)$. In this article, $c$ and $v$ are assumed to be functions on the torus, $\mathbb{T}:=\mathbb{R}/\mathbb{Z}$,  with analytic extension through a band $\abs{\im(z)} \leq \delta$.  As usual, $v$ is real-valued which makes $H_{\theta;\beta}$ a bounded self adjoint operator. We assume $c \not \equiv 0$.

Operators of the form (\ref{eq_hamiltonian}) arise in a tight-binding description of a crystal layer subject to an external magnetic field 
of flux $\beta$ perpendicular to the lattice plane \cite{T,U}. In this context the functions $c$ and $v$ reflect the lattice geometry as well as interactions between the nuclei in the crystal; $\theta$ represents a (random) quasi-momentum. 

A prominent example and the main motivation for this paper is extended Harper's model. Here, the electron is allowed to hop  between nearest and next nearest neighboring lattice sites expressed through the couplings $\lambda_2$ and $\lambda_1, \lambda_3$, respectively, 
\begin{equation} \label{eq_hamiltonian1}
c_{\lambda}(x) := \lambda_{3} \mathrm{e}^{-2\pi i (x+\frac{\beta}{2})} + \lambda_{2} + \lambda_{1} \mathrm{e}^{2 \pi i (x+\frac{\beta}{2})} ~\mbox{,}
~ v(x)  := 2 \cos(2 \pi x) ~\mbox{.}
\end{equation}
We set $\lambda:=(\lambda_1, \lambda_2, \lambda_3)$ to simplify notation. The model is illustrated in Fig. \ref{figure_2}. Without loss of generality one may assume  $0\leq \lambda_{2} ~\mbox{,} ~0 \leq \lambda_{1} + \lambda_{3}$ and at least one of $\lambda_{1} \mbox{,} ~\lambda_{2} \mbox{,} ~\lambda_{3}$ 
to be positive.

\begin{figure} \label{figure_2}
\includegraphics[width=0.5\textwidth]{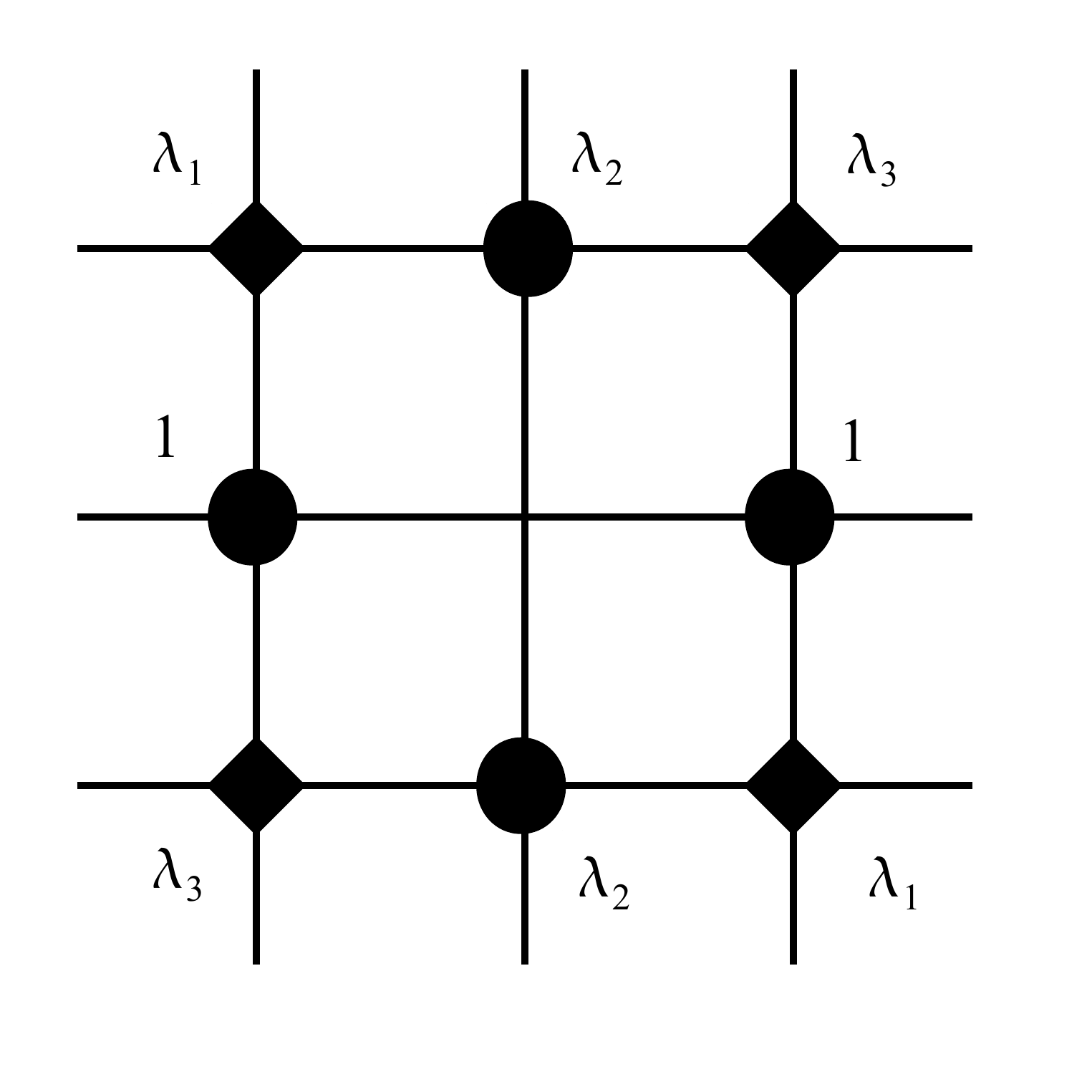}
\caption{Extended Harper's model takes into account both nearest ($\lambda_2$) and next nearest neighbor interaction ($\lambda_1$, $\lambda_3$). Relevant special cases are Harper's model ($\lambda_1=\lambda_3=0$, also known as ``almost Mathieu'') and the triangular lattice (one of $\lambda_1$, $\lambda_3$ equals zero). Extended Harper's model as given here was obtained by renormalization, reducing the number of coupling parameters from four to three.}
\end{figure}

The model is sufficiently general to incorporate both rectangular and triangular (one of $\lambda_1, \lambda_3$ zero) lattice geometries. We also mention that $\lambda_1 = \lambda_3 = 0$ produces the well known almost Mathieu operator (in physics literature also known as Harper's model) with coupling $\lambda_2^{-1}$. Contrary to Harper's model, not much is known about its extension given in (\ref{eq_hamiltonian1}).

Analyzing the solutions to the generalized eigenvalue problem $H_{\theta;\beta} \psi = E \psi$ from a dynamical systems point of view leads to consideration
of the analytic cocycle $(\beta, A^{E}(x))$ where
\begin{equation} \label{eq_deftransfer1}
A^{E}(x) := \begin{pmatrix} E - v(x) & -\overline{c}(x-\beta) \\ c(x) & 0 \end{pmatrix} ~\mbox{.}
\end{equation}
We will call $(\beta, A^E)$ a {\em{Jacobi-cocycle}}, generalizing the notion of Schr\"odinger cocyle associated with the special case $c=1$. A precise definition of (analytic) cocycles is given in Definition \ref{def_anacoc}. 

A complication not present in the study of Schr\"odinger cocycles is the possibility of singularities. A Jacobi cocycle is called {\em{singular}} if for some $x_0 \in \mathbb{T}$, $\det A(x_0) = \overline{c}(x-\beta) c(x) = 0$ ($x_0$ correspondingly is termed a singularity of $(\beta,A^E)$).

By Oseledets' theorem (\cite{AA}; see also Appendix \ref{app_oseledec}), the asymptotic behavior of solutions to $H_{\theta;\beta} \psi = E \psi$ is characterized by the Lyapunov exponent (LE) associated with $(\beta, A^{E})$. In particular, knowledge of the LE is a crucial tool to tackle the spectral analysis for the ergodic operators $H_{\theta;\beta}$. 

Our main result resolves the open problem of obtaining a complete description of the LE of extended Harper's equation as a function of the coupling $\lambda$. Earlier attempts to compute the LE, both heuristic and rigorous, were based on an underlying symmetry of the Hamiltonian, known as duality \cite{D,E,G,H,EE}. Relying on this symmetry, however, precludes the analysis of a certain significant region of couplings.

With respect to duality, the parameter space splits into (see also Fig. \ref{figure_1}):
\begin{description}
\item[region I] $0 \leq \lambda_{1}+\lambda_{3} \leq 1, ~0 \leq \lambda_{2} \leq 1$ ~\mbox{,}
\item[region II] $0 \leq \lambda_{1}+\lambda_{3} \leq \lambda_{2}, ~1 \leq \lambda_{2} $ ~\mbox{,}
\item[region III] $\max\{1,\lambda_{2}\} \leq \lambda_{1}+\lambda_{3}$ ~\mbox{.}
\end{description}
According to the action of the duality transformation, regions I and II are dual regions whereas region III is self-dual. For a precise meaning of duality for extended Harper's model we refer the reader to Appendix \ref{app_dualityconj}. A more general perspective on duality is given in \cite{D}. We reiterate that the duality based approach {\em{a priori}} excludes the self dual region.

\begin{figure} \label{figure_1}
\includegraphics[width=0.5\textwidth]{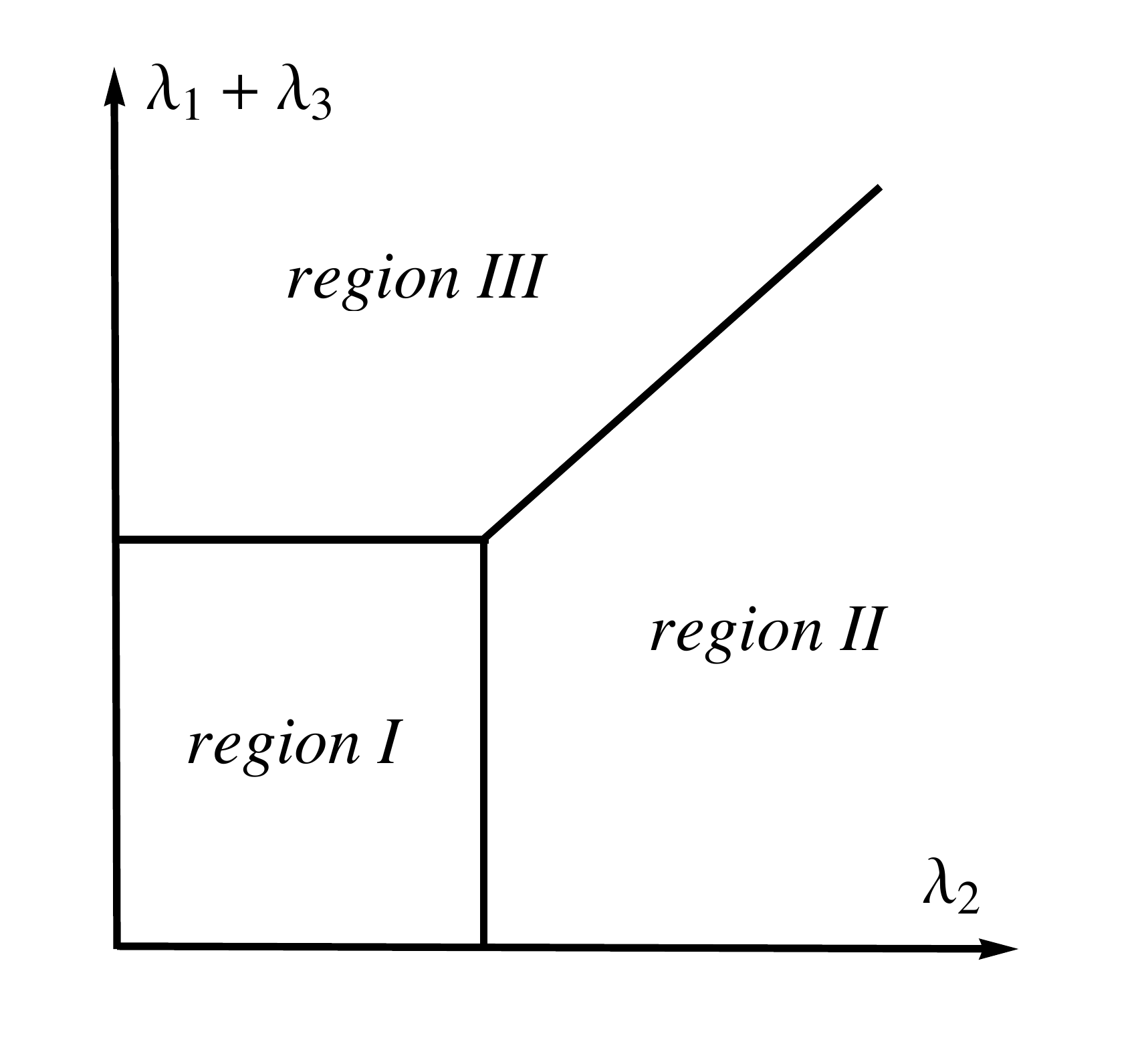}
\caption{Partitioning of the parameter space into regions I - III. So far duality based methods precluded a study of the self dual region III.}
\end{figure}

So far the only rigorous results on extended Harper's equation have
been obtained in the dual regime (regions I \& II) for $\beta$ Diophantine (see (\ref{eq_diophantine}) for a definition) \cite{D, E, EE}. 

Ref. \cite{EE} quantifies the LE in region I for Diophantine $\beta$.  This formula for the LE in region I was obtained based on \cite{E} which proves (spectral) localization for Diophantine $\beta$ in the interior of region I. Consequently, as shown in \cite{E}, duality forces the LE in the interior of the dual region II to be zero.

We mention that these results for the dual regime confirm non-rigorous computations by Thouless \cite{G,H}.  Since Thouless' considerations were also based on duality, they could not provide insight in region III either. 

As a consequence of their heuristic analysis of the Lyapunov exponent, in \cite{G} Thouless and Hahn concluded localization within region I, and extended states for region
II. Note that regions I and II can be viewed as extensions of the corresponding (metallic and insulator) regions of the almost Mathieu operator ($\lambda_1=\lambda_3=0$). 

The self-dual region III, however, where next-nearest-neighbor interaction dominates, does not have analogues in the almost Mathieu family.  In that region, the
authors of \cite{G} had to rely on numerical studies carried out for the case $\lambda_1 = \lambda_3,$ which indicated singular continuous spectrum \cite{G,H}. We note that $\lambda_1 = \lambda_3$ corresponds to the most physically relevant case of isotropic next-nearest-neighbor interaction, yet it turns out to be
the most difficult one to study rigorously, because of the position of singularities of the associated Jacobi cocycle (see Observation \ref{obs_cfun}).

The main achievement of the present article is to overcome the problem of analyzing the self-dual region of extended Harper's model. We prove the following:
\begin{theorem} \label{thm_mainresult}
Fix an irrational frequency $\beta$. Then the Lyapunov exponent on the spectrum is zero within both region II and III. In region I it is given by the formula,
\begin{equation} \label{eq_thoulessform}
\begin{cases}\log \left(  \dfrac{1+\sqrt{1 - 4\lambda_{1} \lambda_{3}}}{2 \lambda_{1}}\right) & \mbox{, if} ~\lambda_{1} \geq \lambda_{3}, ~\lambda_{2} \leq \lambda_{3} + \lambda_{1} ~\mbox{,}  \\
\log \left(  \dfrac{1+\sqrt{1 - 4\lambda_{1} \lambda_{3}}}{2 \lambda_{3}}\right) & \mbox{, if} ~\lambda_{3} \geq \lambda_{1}, ~\lambda_{2} \leq \lambda_{3} + \lambda_{1} ~\mbox{,}  \\
\log \left(  \dfrac{1+\sqrt{1 - 4\lambda_{1} \lambda_{3}}}{\lambda_{2} + \sqrt{\lambda_{2}^{2} - 4 \lambda_{1} \lambda_{3}}}\right) & ~\mbox{, if} ~\lambda_{2} \geq \lambda_{3} + \lambda_{1} ~\mbox{.} 
\end{cases}
\end{equation}
\end{theorem}
We emphasize that Theorem \ref{thm_mainresult} holds for all {\em{irrational}} frequencies. 

Our strategy uses  complexification of the analytic cocycle $(\beta, A^{E})$ introduced in (\ref{eq_deftransfer1}). For all $\lambda$, where $(\beta, A^E)$ does not possess singularities, our method is based on Avila's global theory of analytic one-frequency Schr\"odinger operators \cite{B}, more precisely its extension to the GL(2,$\mathbb{C}$) case, that we develop in Sec. \ref{sec_jacobico}.

A complication, however, arises in extended Harper's equation due to zeros of $\det A^{E}(x)$ relating to zeros
of $c_{\lambda}(x)$. Singularity of the cocycle $(\beta,A^{E})$ constitutes an important difference to analytic Schr\"odinger
cocycles. \footnote{We note that the holomorphic  extension $A(x+i\epsilon)$ of the transfer matrix to all of $\mathbb{C}$ will encounter zeros for 
some {\bf{some non-zero}} $\epsilon$  for any $\lambda$ with $(\lambda_1,\lambda_3)\neq (0,0),$ (see Observation \ref{obs_cfun}), however, as we will
argue, these are much easier to deal with.}

An important ingredient that allows us to both overcome singularities of the cocycle as well as extend our results to all irrational $\beta$ is continuity of the Lyapunov exponent in the energy, $\lambda$, and $\beta$.  As will be shown, this is only a special case of a more general continuity property of the Lyapunov exponent, valid for the class of analytic cocycles. 

To this end denote by $\mathcal{C}^\omega(\mathbb{T},M_2(\mathbb{C}))$ the class of 1-periodic functions on the real line, with analytic extension to {\em{some}} band, $\abs{\im{z}} \leq \delta$, attaining values in the complex 2$\times$2 matrices, $M_2(\mathbb{C})$. The space $\mathcal{C}^\omega(\mathbb{T},M_2(\mathbb{C}))$ is topologized by an inductive limit topology described at the beginning of Sec. \ref{sec_jacobico}. 

\begin{theorem} \label{thm_conti_sing}
Given $\beta$ irrational, the Lyapunov-exponent, $L(\beta + ., .): \mathbb{T} \times \mathcal{C}^\omega(\mathbb{T},M_2(\mathbb{C})) \to \mathbb{R}$ is jointly continuous.
\end{theorem}
\begin{remark} 
\begin{itemize}
\item[(i)] It is important that $\beta$ is irrational, since at rational $\beta$ the theorem is known to be false even for the $SL(2,\mathbb{C})$-case (see e.g. Remark 5 in \cite{B} for a counter-example).
\item[(ii)] If $D \equiv 0$, $L(\beta, D)=-\infty$ for any $\beta \in \mathbb{R}$. \label{rem_defletriv}
\end{itemize}
\end{remark}

Theorem \ref{thm_conti_sing} is false when reducing the degree of regularity to even smooth cocycles. Furthermore, let $\mathbb{N}_0:=\mathbb{N}\cup\{0\}$. For $k \in \mathbb{N}_0 \cup \{\infty\}$, let $\mathcal{C}^k(\mathbb{T},M_2(\mathbb{C}))$ denote the class of 1-periodic $M_2(\mathbb{C})$-valued functions on the real line,  topologized by the complete metric \footnote{For $k \in \mathbb{N}_0$ the metric is in fact derived from a norm, which however is {\em{not}} true for $k=\infty$.}
\begin{equation} \label{eq_metricsck}
\rho_k(D,E):= \begin{cases} \max_{0\leq j \leq k} \max_{x \in \mathbb{T}} \norm{\partial^j (D(x) - E(x))} & \mbox{,} ~k \in \mathbb{N}_0 ~\mbox{,} \\
                                                \sum_{n=0}^{\infty} \frac{1}{2^n} \frac{\rho_n(D,E)}{1+\rho_n(D,E)} & \mbox{,} ~k = \infty ~\mbox{,} \end{cases}
\end{equation}
where $\norm{.}$ is any fixed matrix norm %and for all $n\in \mathbb{N}_0$, $0<b_n$ with $\sum_{n\in \mathbb{N}_0} b_n< \infty$
. Here, $k=0$ is simply the continuous maps, formally $\partial^0 = \mathrm{id}$, the identity operation on functions. For $k \in \mathbb{N}$, $\partial^k D$ denotes the $kth$ derivative taken componentwise.  We restrict to the subset
\begin{equation}
\mathcal{L}^k:=\{ D \in \mathcal{C}^k(\mathbb{T},M_2(\mathbb{C})): \log \norm{D} \in L^1(\mathbb{T},\ud x) \} ~\mbox{,}
\end{equation}
where the LE is finite (see (\ref{eq_defle_limit})). By Fatou's Lemma $\mathcal{L}^k$ is open in $\mathcal{C}^k(\mathbb{T},M_2(\mathbb{C}))$. Then, we can claim:
\begin{theorem} \label{thm_conti_sing_opt}
Theorem \ref{thm_conti_sing} is optimal, i.e. $L(\beta+.,.) : \mathbb{T} \times \mathcal{L}^\infty \to \mathbb{R}$ is discontinuous. Moreover, for any $k \in \mathbb{N}_0$, $L(\beta + ., .): \mathbb{T} \times \mathcal{C}^\omega \to \mathbb{R}$ is discontinuous in the topology of $\mathbb{T} \times\mathcal{C}^k$. 
\end{theorem}
\begin{remark}
\begin{itemize}
\item[(i)] Certainly this implies discontinuity in $\mathcal{C}^k(\mathbb{T},M_2(\mathbb{C})), k\in \mathbb{N}\cup\{\infty\}$.
\item[(ii)] Wang and You \cite{DD} have recently obtained Theorem \ref{thm_conti_sing_opt} for $SL(2,\mathbb{C})$-cocycles  by explicitly producing subtle counter-examples for any degree of regularity $\mathcal{C}^k(\mathbb{T}, SL(2,\mathbb{C}))$, $k \in \mathbb{N} \cup \{ \infty \}$. In case of $M_2(\mathbb{C})$-cocycles, however, the construction can be done using elementary tools of harmonic analysis only, which is what is done here.
\item[(iii)] Discontinuity for $k=0$ is well known. In \cite{GG}, Furman proved that for irrational $\beta$, the LE is discontinuous at every $(\beta, D)$ where the limit (\ref{eq_ledef2}) is not uniform. More recently, in \cite{GGG}, Bochi proved that cocycles with zero LE are dense in $\mathcal{C}^0(\mathbb{T},SL(2,\mathbb{C}))$.
\end{itemize}
\end{remark}

Theorem \ref{thm_conti_sing} was preceded by a result in \cite{CC} which had already obtained continuity of the LE for singular analytic cocycles but only under a Diophantine condition. The present formulation removes this limitation. 

The main accomplishment here, however, is to also prove continuity in the frequency at any irrational $\beta$, which allows to consider rational approximates of the frequency. Such rational approximations are particularly useful for practical purposes since the resulting dynamical systems are periodic. For instance, continuity in this sense was the key assumption that Thouless used to compute the LE of extended Harper's model in region I.

It is noteworthy that Theorem \ref{thm_conti_sing} arises as a consequence of ideas developed in Sec. \ref{sec_reglyap} and illustrates the power of Avila's global theory approach \cite{B}. The proof given here not only removes the problem of our earlier work to deal with Liouvillean frequencies but is also surprisingly simple. 
The key idea underlying this argument was communicated to us by Artur Avila (see more in the Acknowledgement).

We organize the paper as follows. Sec. \ref{sec_reglyap} forms the technical heart of our strategy to determine the Lyapunov exponent of extended Harper's equation. Based on analyticity of the cocycle $(\beta,A^{E})$ we consider the complex extension $A^{E}(x+i \epsilon)=:A_\epsilon^E(x)$, $\epsilon \in \mathbb{R}$. The underlying idea is to achieve ``{\em{almost constant cocycles}}'' by considering the limits $\epsilon \to \pm \infty$.

In Sec. \ref{sec_jacobico}, we extend Avila's global theory of analytic Schr\"odinger (or more generally, SL(2,$\mathbb{C}$)) cocycles to the non-singular Jacobi (or GL(2,$\mathbb{C}$)) case. In particular, in Theorems \ref{thm_quantacc}, \ref{thm_unifhyp} and \ref{coro_regspec} we analyze the dependence of the Lyapunov exponent of the complexified cocycle on $\epsilon$, eventually enabling us to extrapolate to $\epsilon=0$. 

As an immediate application,  in Sec. \ref{sec_cont} we establish continuity of the LE for singular analytic cocycles as stated in Theorem \ref{thm_conti_sing}. This is the key ingredient to deal with singularities of the cocycle. In the same section we also prove optimality of this result (Theorem \ref{thm_conti_sing_opt}).

Having summarized some simple facts about extended Harper's model in Sec. \ref{sec_ehm}, Section \ref{sec_asympt} obtains asymptotic expressions of its Lyapunov exponent. Sec.  \ref{sec_whatwecansay} carries out the extrapolation process leading to Theorem \ref{thm_mainresult}. In fact, as a result of our analysis we obtain the LE for all $\epsilon$ (Theorem \ref{prop_LE}).

An interesting property of the LE of extended Harper's equation is its independence of $\lambda_2$ in region III (see Fig. \ref{figure_1}). This symmetry is already revealed in the asymptotic analysis of Sec. \ref{sec_asympt} and hence inspires an alternative proof of Theorem \ref{thm_mainresult} which we give in Sec. \ref{sec_alternative}.

We mention that the strategy we develop to obtain the LE for extended Harper's equation does not rely on the details of the functions $c,v$ in (\ref{eq_hamiltonian1}); our method is hence applicable to a general {\em{singular}} quasi-periodic Jacobi operator with analytic coefficients.

Finally, in view of future applications to other Jacobi operators, in Sec. \ref{sec_concl} we point out differences and peculiarities of the Jacobi case as compared to Schr\"odinger cocycles which relate to possible zeros of $\det (A^E(x))$.  As we will see, extended Harper's model provides useful examples and counterexamples. 

In particular, we address the question of {\em{almost reducibility}} which, for Schr\"odinger cocycles, is conjectured to provide a necessary criterion for (purely) absolutely continuous spectrum in terms of the $\epsilon$ dependence of $L(\beta,A_\epsilon^E)$ near $\epsilon=0$ (``almost reducibility conjecture (ARC)'' \cite{B,V,W,X}). Extended Harper's equation provides a nice example that the ARC is false for the general Jacobi case due to possible zeros of $\det(A^{E}(x))$ on $\mathbb{T}$.

{\bf Acknowledgement.} We are grateful to Artur Avila for his remarks on an earlier version of this paper where, among other things,  he essentially provided a simple proof of continuity of the Lyapunov exponent of singular cocycles for all frequencies once such continuity for the Diophantine case is
established, using the ideas of \cite{B}.
 It turned out the same idea 
could be used to provide a simple proof of joint continuity (presented here), significantly simplifying our original approach.  Additionally, our proof of continuity for the case of identically vanishing determinant follows his suggestions as well.
We also thank Anton Gorodetski for useful discussions during the preparation of this manuscript.
\section{Complexified cocycles} \label{sec_reglyap}

\subsection{Jacobi cocycles - General results} \label{sec_jacobico}

We start by considering a general quasi periodic Jacobi operator with analytic coefficients as given in (\ref{eq_hamiltonian}). We will be dealing with analytic functions on $\mathbb{T}$. To this end, some topological remarks will come handy.

Given a Banach space $(\mathfrak{X},\Vert . \Vert_{\mathfrak{X}})$, let $\mathcal{C}_\delta^\omega(\mathbb{T}, \mathfrak{X})$ be the space of $\mathfrak{X}$-valued functions on $\mathbb{T}$ with holomorphic extension to a neighborhood of $\abs{\im{z}} \leq \delta$, $\delta>0$. We shall denote the band $\abs{\im{z}} \leq \delta$ about $\mathbb{T}$ by $\mathbb{T}_\delta$. The set $\mathcal{C}_\delta^\omega(\mathbb{T}, \mathfrak{X})$ naturally becomes a Banach space in its own right when equipping it with the norm $\Vert X \Vert_\delta := \sup_{\abs{\im{z}}\leq \delta} \Vert X(z) \Vert_{\mathfrak{X}}$, for $X(.) \in \mathcal{C}_\delta^\omega(\mathbb{T}, \mathfrak{X})$. 

To obtain statements independent of $\delta$, we consider $\mathcal{C}^\omega(\mathbb{T};\mathfrak{X}):=\cup_{\delta>0} \mathcal{C}_\delta^\omega(\mathbb{T}, \mathfrak{X})$. The appropriate topology for $\mathcal{C}^\omega(\mathbb{T};\mathfrak{X})$ is the {\em{inductive limit topology}} induced by $\Vert . \Vert_\delta$. In this topology, convergence of a sequence $X_n \to X$ is equivalent to existence of some $\delta>0$ such that $X_n \in \mathcal{C}_\delta^\omega(\mathbb{T}, \mathfrak{X})$ holds eventually and $\Vert X_n - X \Vert_\delta \to 0$ as $n \to \infty$. 
%$\mathcal{C}^\omega(\mathbb{T};\mathfrak{X})$ thus defined forms a complete, locally convex topological vector space %possessing the Heine-Borel property.

If not specifically stated the sets $\mathcal{C}_\delta^\omega(\mathbb{T}, \mathfrak{X})$ and $\mathcal{C}^\omega(\mathbb{T}, \mathfrak{X})$ will always be understood as  topological spaces with respect to the above mentioned topologies.

For $c(x) \neq 0$, we define the transfer matrix associated with the equation $H_{\theta;\beta} \psi = E \psi$,
\begin{eqnarray} \label{eq_deftransfer}
B^{E}(x) := \dfrac{1}{c(x)} \begin{pmatrix} E - v(x) & -\overline{c}(x-\beta) \\ c(x) & 0 \end{pmatrix} =: \dfrac{1}{c(x)} A^{E}(x)~\mbox{,}
\end{eqnarray}
where solutions to the Schr\"odinger equation are obtained iteratively by
\begin{eqnarray}
& \begin{pmatrix} \psi_{n} \\ \psi_{n-1} \end{pmatrix} = B^{E;n}(\beta,\theta) \begin{pmatrix} \psi_{0} \\ \psi_{-1} \end{pmatrix}  ~\mbox{,} \nonumber \\
& B^{E;n}(\beta,\theta) := B^{E}(\theta + \beta (n-1)) \dots  B^{E}(\theta) ~\mbox{,} \nonumber\\
& B^{E;n}(\beta,\theta) := B^{E;n}(\beta,\theta - n \beta)^{-1} ~\mbox{,} ~n\geq 1 ~\mbox{.} \label{eq_iterates}
\end{eqnarray}

Equation (\ref{eq_iterates}) has a convenient dynamical formulation based on the following definition:
\begin{definition} \label{def_anacoc}
\begin{itemize}
\item[(i)] For $\beta \in \mathbb{R}$ and $D: \mathbb{T} \to M_{2}(\mathbb{C})$ Borel-measurable satisfying 
\begin{equation} \label{eq_anaco}
\int_{\mathbb{T}} \log_+{\norm{D(x)}} \ud x < \infty ~\mbox{,}
\end{equation}
we call the pair $(\beta,D(x))$ a {\em{cocycle}} understood as a linear skew-product acting on  $\mathbb{T} \times \mathbb{C}^{2}$ defined by $(x,v) \mapsto \left(x+\beta, D(x) v\right)$. In (\ref{eq_anaco}), $\norm{.}$ may be {\em{any}} norm on $M_2(\mathbb{C})$. If $D \in \mathcal{C}^\alpha(\mathbb{T},M_2(\mathbb{C}))$, $\alpha \in \mathbb{N}_0 \cup \{\infty,\omega\}$, $(\beta,D(x))$ is referred to as $\mathcal{C}^\alpha$-cocyle (also {\em{analytic cocycle}} if $\alpha=\omega$). 
\item[(ii)] A $\mathcal{C}^\alpha$-cocycle $(\beta,D(x))$ is called {\em{singular}} if $\det D(x_0)=0$ for some $x_0 \in \mathbb{T}$, in which case $x_0$ is referred to as singularity of the cocycle $(\beta,D(x))$.
\end{itemize}
\end{definition}
\begin{remark}
For any $D \in \mathcal{C}^\omega(\mathbb{T},M_2(\mathbb{C}))\setminus \{0\}$, analyticity guarantees $\log{\norm{D}} \in L^1(\mathbb{T}, \ud x)$ (for a simple argument see e.g. the proof of Lemma 2.9 in \cite{CC}). 
\end{remark}

Since the functions $c,v$ extend analytically to some band $\abs{\im z}\leq \delta$, we may consider the complexified transfer matrix $B_{\epsilon}^{E}(x):=B^{E}(x+i\epsilon)$ for $\abs{\epsilon} \leq \delta$. Set
\begin{equation} \label{eq_Ifun}
I_{\epsilon}(c) := \int_{\mathbb{T}} \log \abs{c(x+i\epsilon)} \ud x = \frac{1}{2} \int_{\mathbb{T}} \log{\abs{\det{A_{\epsilon}^{E}}(x)}} \ud x ~\mbox{,}
\end{equation}
and write $I_{0}(c) =: I(c)$ to simplify notation.

Given a cocycle $(\beta, D)$ with $\beta$ irrational, the Lyapunov exponent is defined by
\begin{eqnarray} \label{eq_defle_limit}
L(\beta, D) &:=& \lim_{n\to\infty} \frac{1}{n} \int_{\mathbb{T}} \log \norm{D^{(n)}(x)} \ud x =  \inf_{n \in \mathbb{N}} \frac{1}{n} \int_{\mathbb{T}} \log \norm{D^{(n)}(x)} \label{eq_ledef1}  \\
& = & \lim_{n \to \infty} \frac{1}{n} \log\norm{D^{(n)}(x)} ~\mbox{,} \label{eq_ledef2}
\end{eqnarray}
where 
\begin{equation} \label{eq_defiteratesgen}
D^{(n)}(x):=D(x+(n-1)\beta)\dots D(x)
\end{equation}
is the analogue of the $n$-step transfer matrix in (\ref{eq_iterates}). Existence and a.e. independence of the limit in (\ref{eq_ledef2}) follows by the sub-additive ergodic theorem.

In particular, for the cocycles $(\beta, A^E)$ and $(\beta, B^E)$, equation (\ref{eq_deftransfer}) implies the following relation 
\begin{equation} \label{eq_relAB}
L(\beta,B_{\epsilon}^{E}) = L(\beta, A_{\epsilon}^{E}) - I_{\epsilon}(\lambda)~\mbox{,}
\end{equation}
which allows us to focus on the {\em{analytic}} cocycle $(\beta, A^E)$.

In view of Theorem \ref{thm_conti_sing}, notice that by the sub-additive ergodic theorem the limit on the right hand side of (\ref{eq_ledef2}) still exists a.e. (is $L^1$ and invariant under rotations by $\beta$) even if $\beta$ is rational but generally this limit will depend on $x$ (and thus will not be equal to (\ref{eq_ledef1})). 

For {\em{rational}} $\beta$, $\beta = \frac{p}{q}$ with $(p,q)=1$, we define the LE as given in (\ref{eq_ledef1}) and note that by dominated convergence,
\begin{equation} \label{eq_le_rat}
L(\frac{p}{q}, D) = \frac{1}{q} \int_\mathbb{T} \log \rho\left( D^{(q)}(x) \right) \ud x ~\mbox{.}
\end{equation}
Here and later, $\rho(A)$ denotes the spectral radius of a matrix $A$.

Following, we denote the spectrum of $H_{\theta;\beta}$ by $\Sigma$ and its a.s. components by $\Sigma_{pp}$, $\Sigma_{ac}$, and $\Sigma_{sc}$, respectively.
Our analysis of $L(\beta, B^{E})$, for $E \in \Sigma$, is based on the complexification of the cocycle $(\beta, A^{E}(x))$. 

To this end, let $(\beta,D(x))$ be a fixed analytic cocycle. We introduce the acceleration,
\begin{equation} \label{def_acc}
\omega(\beta, D; \epsilon) := \frac{1}{2\pi} \lim_{h \to 0+} \dfrac{L(\beta, D_{\epsilon+h}) - L(\beta, D_{\epsilon})}{h} ~\mbox{.}
\end{equation}
The acceleration was first introduced in \cite{B} for analytic $SL(2,\mathbb{C})$-cocycles. Existence of $\omega(\beta, D; \epsilon)$ is a consequence of convexity of $L(\beta, D_{\epsilon})$ w.r.t. $\epsilon$ which is still true even if $(\beta,D)$ is singular. Albeit simple, convexity of $L(\beta, D_{\epsilon})$ w.r.t. $\epsilon$ for singular analytic cocycles will be shown to have far-reaching consequences in what is to come. We include a brief argument in Appendix \ref{app_1} (see Proposition \ref{prop_convlyap}).

The next two theorems, Theorem \ref{thm_quantacc} and  \ref{thm_unifhyp}, extend statements proven by Avila for analytic cocyles with SL(2,$\mathbb{C}$)-transfer matrices 
\cite{B}. Avila's proof uses continuity of the Lyapunov exponent for Schr\"odinger cocycles established in \cite{C} (see Theorem \ref{thm_bj} below). 

It is straightforward to extend these statements to analytic cocycles whose determinants are bounded away from zero on a strip, making use of the following Lemma which we believe to be well known. For the reader's convenience, we provide proof in Appendix \ref{app_lemmaproof}.
\begin{lemma} \label{lem_root}
Let $f \in \mathcal{C}^\omega(\mathbb{T};\mathbb{C})$ with $\min_{x \in \mathbb{T}} \abs{f(x)} > 0$. Then, there exists $g \in \mathcal{C}^\omega(\mathbb{R}/2\mathbb{Z};\mathbb{C})$ such that $g^2 = f$.
\end{lemma} 
\begin{remark}
\begin{itemize}
\item[(i)] From basic complex analysis it is clear that $\sqrt{f}$ can be defined holomorphically in a neighborhood of $\mathbb{R}$, however it is not a priori obvious that this also yields a periodic function (which in general is 2- instead of 1-periodic).
\item[(ii)] The same proof shows that $f^{\frac{p}{q}}$, can be defined in $\mathcal{C}^\omega(\mathbb{R}/q \mathbb{Z}; \mathbb{C})$.
\end{itemize}
\end{remark}

As a consequence of Lemma \ref{lem_root}, Avila's results on analytic $SL(2,\mathbb{C})$-cocycles carry over to a {\em{non-singular}} analytic cocycle $(\beta, D)$ upon consideration of the ``renormalized'' $SL(2,\mathbb{C})$-cocycle,
\begin{equation} \label{eq_Dpr}
D^{\prime}:=\frac{D}{\sqrt{\det{D}}} ~\mbox{,} ~D^\prime \in \mathcal{C}^\omega(\mathbb{R}/2 \mathbb{Z}; M_2(\mathbb{C})) ~\mbox{.}
\end{equation}
In this context, it is useful to notice that a given matrix-valued function $D: \mathbb{T} \to M_2(\mathbb{C})$ satisfying the hypotheses of Definition \ref{def_anacoc} (i)  may also be considered as a function on $\mathbb{R}/2 \mathbb{Z}$, in which case the Lyapunov exponents of the respective cocycles (for a fixed irrational $\beta$) are related by a factor 2.

\begin{theorem}[Quantization of acceleration] \label{thm_quantacc}
Consider an analytic cocycle $(\beta,D(x))$ where $\beta$ is irrational and $\det D(x)$ bounded away from zero on a strip $\mathbb{T}_\delta$. Then, the acceleration on $\mathbb{T}_\delta$ is integer-valued.
\end{theorem}
Using Lemma \ref{lem_root}, we reduce to the original result for SL(2,$\mathbb{C})$-cocycles as stated in \cite{B}, noticing that
\begin{equation} \label{eq_qa1}
L(\beta, D) = L(\beta, D^\prime) + \frac{1}{2} \int_{\mathbb{T}} \log \abs{\det(D)} \ud x ~\mbox{.}
\end{equation}

What is left to analyze is the second term on the right hand side of (\ref{eq_qa1}). We mention that this integral can be recast as the LE of a diagonal cocycle:
\begin{equation} \label{eq_qa2}
\frac{1}{2} \int_{\mathbb{T}} \log\abs{\det(D)} \ud x = L\left(\beta, \begin{pmatrix} \det D(z)  & 0 \\ 0 & \det D(z) \end{pmatrix}\right) ~\mbox{.} 
\end{equation}

In this case, however, the content of Theorem \ref{thm_quantacc} may easily be checked directly, which is the subject of the following Lemma:
\begin{lemma} \label{lem_quantaccbaby}
\begin{itemize}
\item[(i)] For $c \in \mathcal{C}_\delta^\omega(\mathbb{T},\mathbb{C})$ with $\min_{x\in\mathbb{T}_\delta} \abs{c(x)} > 0$ the function $I_\epsilon(c)$ defined in (\ref{eq_Ifun}) for $\epsilon \in [-\delta,\delta]$  is affine in $\epsilon$ with derivative in $2 \pi \mathbb{Z}$. 
\item[(ii)] If $c \in \mathcal{C}_\delta^\omega(\mathbb{T},\mathbb{C})\setminus \{0\}$ is not bounded away from zero, $I_\epsilon(c)$ is a piecewise linear, convex function in $\epsilon$ with right derivatives in $2 \pi \mathbb{Z}$. 
\end{itemize}
\end{lemma}
\begin{proof}
\begin{itemize}
\item[(i)]
Since $c(z)$ is holomorphic with no zeros on $\abs{\im{z}} \leq \delta$,  $\log{\abs{c(z)}}$ is
harmonic on the same strip. Thus one computes,
\begin{eqnarray} \label{eq_remark}
 \dfrac{\ud^{2}}{\ud \epsilon^{2}}\int_{\mathbb{T}} \log{\abs{c(x+i\epsilon)}} \ud x & = & \int_{\mathbb{T}} \dfrac{\partial^{2}}{\partial \epsilon^{2}}\log{\abs{c(x+i\epsilon)}} \ud x \nonumber \\
 & = & - \int_{\mathbb{T}} \dfrac{\partial^{2}}{\partial x^{2}}\log{\abs{c(x+i\epsilon)}} \ud x = 0 ~\mbox{,}
\end{eqnarray}
from which we conclude that $I_\epsilon(c)$ is affine on $\abs{\epsilon} \leq \delta$ as claimed.

To show that the derivative is in $2 \pi \mathbb{Z}$, we first consider the case when $c$ is a trigonometric polynomial, i.e. $c(z) = \sum_{k = -N}^{N}  c_k z^k$ where $z:=\mathrm{e}^{2 \pi i (x+i\epsilon)}$. Since,
\begin{equation}
\log\abs{c(z)} = -N \log \abs{z} + \log \left\vert \sum_{k=0}^{2N} a_{k-N} z^k \right\vert ~\mbox{,}
\end{equation}
it suffices to establish the claim for $c(z) = \sum_{k=0}^{N} c_k z^k$.

Denote by $\mathcal{N}(c; K)$ the number of zeros of $c$ (counting multiplicity) on a compact subset $K$ of $\mathbb{C}$ and let $\mathfrak{Z}(c; K)$ denote the associated zero set. Since $c$ is bounded from zero on $\mathbb{T}_\delta$, we have 
\begin{equation}
\mathcal{N}\left(c;\overline{D(0,\mathrm{e}^{-2 \pi \epsilon})}\right) = \mathcal{N}\left(c; \overline{D(0,\mathrm{e}^{-2\pi \delta})}\right) ~\mbox{,}
\end{equation}
whenever $0 \leq \abs{\epsilon} \leq \delta$.

Using Jensen's formula we thus conclude for $0 \leq \abs{\epsilon} \leq \delta$,
\begin{equation}
\int_{\mathbb{T}} \log\abs{c(x+i\epsilon)} \ud x = - 2 \pi \epsilon \mathcal{N}(c;\{0\}) + \sum_{z \in \mathfrak{Z}(c; D(0, \mathrm{e}^{-2 \pi \delta})\setminus \{0\})} n(z) 
\log \left\vert  \dfrac{\mathrm{e}^{-2\pi \epsilon}}{z}  \right\vert + d ~\mbox{,}
\end{equation}
where $d=d(c) \in \mathbb{C}$ and $n(z)$ is the multiplicity of $z$ if $z$ is a root of $c$ and defined zero otherwise.

For a general $c$, uniformly approximate $c$ on $\mathbb{T}_\delta$ by trigonometric polynomials $c_n$. As shown in \cite{CC} (see also Remark \ref{rem_contlemma}), $I_\epsilon(c_n) \to I_\epsilon(c)$ uniformly on $[-\delta,\delta]$. 

In summary we obtain lines, $I_\epsilon(c_n)$, with slopes in $2 \pi \mathbb{Z}$ converging uniformly on $[-\delta,\delta]$ to the line $I_\epsilon(c)$, whence forcing the derivative of the limit to be in $2 \pi \mathbb{Z}$ (see also Fact \ref{fact_conv}).
\item[(ii)] To prove statement (ii), factorize $c$ on $\mathbb{T}_\delta$ according to its roots
\begin{equation}
c(x) = h(x) \prod_{j=1}^{n} \left( \mathrm{e}^{2 \pi x} - \mathrm{e}^{2 \pi i (x_j + i \epsilon_j)} \right)^{n_j} ~\mbox{,}
\end{equation}
where $n_j$ is the multiplicity of the $j$th root and $h$ is zero-free and holomorphic on $\mathbb{T}_\delta$.

Thus, making use of (i),
\begin{eqnarray}
I(c) & = & \sum_{j=1}^{n} n_j \int_{\mathbb{T}} \log \abs{\mathrm{e}^{2\pi i (x+ i \epsilon)} - \mathrm{e}^{2 \pi i (x_j + i \epsilon_j)}} + I(h) \\
      &  = & \sum_{j=1}^{n} n_j \int_{\mathbb{T}} \log \abs{\mathrm{e}^{2\pi i (x+ i \epsilon)} - \mathrm{e}^{2 \pi i (x_j + i \epsilon_j)}} + 2 \pi N \epsilon ~\mbox{,} \label{eq_sum}
\end{eqnarray}
for some $N \in \mathbb{Z}$.

Applying Jensen's formula separately to each summand on the right hand side of (\ref{eq_sum}), we conclude
\begin{equation}
\int_{\mathbb{T}} \log \abs{\mathrm{e}^{2\pi i (x+ i \epsilon)} - \mathrm{e}^{2 \pi i (x_j + i \epsilon_j)}} = - 2 \pi \min\{\epsilon, \epsilon_j\} ~\mbox{.} \label{eq_sum1}
\end{equation}

Combining (\ref{eq_sum}) and (\ref{eq_sum1}) yields the claim.
\end{itemize}
\end{proof}

The following special case will turn out to be of relevance for our further development:
\begin{definition} 
\cite{B} If $L(\beta, D_{\epsilon})$ is affine about $\epsilon=0$, the associated cocycle $(\beta,D)$ is referred to as {\em{regular}}.
\end{definition}

Next we explore the relation between regularity and the dynamics induced by the cocycle $(\beta, D)$. To this end we define:
\begin{definition} \label{def_unifhyp}
An (analytic) SL(2,$\mathbb{C}$)-cocycle $(\beta,D)$ is called {\em{uniformly hyperbolic}} if there exist (analytic) maps $s,u: \mathbb{T} \to \mathbb{PC}^{2}$ such that 
\begin{itemize}
\item[(i)] $D(x) u(x) = u(x+\beta)$, $D(x) s(x) = s(x+\beta)$,
\item[(ii)] for $w \in \mathbb{C}^{2}$, $\norm{w}=1$: $\norm{D(x) w}>1$, if $\pi(w) = u(x)$, and $\norm{D(x) w}<1$, if $\pi(w) = s(x)$, $\forall x \in \mathbb{T}$. Here, $\pi$ denotes the canonical projection of $\mathbb{C}^2$ onto $\mathbb{PC}^2$.
\end{itemize}
\end{definition}

Using Lemma \ref{lem_root}, the following is obtained as a mere corollary of Theorem 6 in \cite{B}:
\begin{theorem} \label{thm_unifhyp}
Consider an analytic cocycle $(\beta,D)$ where $\beta$ is irrational and $\det{D}$ bounded away from zero on $\abs{\im{z}} \leq \delta$. 
Assume $L(\beta, D^{\prime}) > 0$. Then the cocycle $(\beta,D^{\prime})$ is regular if and only if it is uniformly hyperbolic.
\end{theorem}

In particular, applying Theorem \ref{thm_unifhyp} to Jacobi operators, we obtain a statement which will be central to compute $L(\beta, B^{E})$ on the spectrum:
\begin{theorem} \label{coro_regspec}
Let $\beta$ irrational and assume $\det A^E(z)$ bounded away from zero for $\abs{\im z} \leq \delta$. If $E \in \Sigma$ with $L(B^{E},\beta) > 0$, the cocycles $(\beta, A^{E})$ and $(\beta, B^{E})$ cannot be regular.
\end{theorem}
We note that 
\begin{equation}
L(\beta, (A_\epsilon^E)^\prime) = L(\beta, B_\epsilon^E) ~\mbox{.}
\end{equation}
The proof of Theorem \ref{coro_regspec} is a consequence of the following two Lemmas and Theorem \ref{thm_unifhyp}. 

\begin{lemma} \label{lemma_1}
Let $\beta$ irrational and assume $\det A^E(z)$ bounded away from zero for $\abs{\im z} \leq \delta$. Then,  $(\beta, (A^{E})^{\prime})$ is regular if
and only if $\left(\beta,A^{E}\right)$ is regular.
\end{lemma}
\begin{proof}
Lemma \ref{lem_quantaccbaby} implies that for $\abs{\epsilon} \leq \delta$ we have
\begin{equation}
L(\beta, (A_{\epsilon}^{E})^{\prime}) = L(\beta, A_{\epsilon}^{E}) + 2 \pi N \epsilon + \Gamma  ~\mbox{,}
\end{equation}
some $\Gamma \in \mathbb{R}$ and $N \in \mathbb{Z}$.
Since an analytic cocycle is not regular if and only if its acceleration has a jump discontinuity at $\epsilon = 0$, we obtain the claim.
\end{proof}

\begin{lemma} \label{lemma_regunif}
Let $\beta$ irrational and assume $\det A^E(z)$ bounded away from zero for $\abs{\im z} \leq \delta$. If $E \in \Sigma$, the cocycles $(\beta, (A_{\epsilon}^{E})^{\prime})$ and $(\beta, B^{E})$ cannot be uniformly hyperbolic.
\end{lemma}
The statement is well known for $(\beta, B^{E})$. The claim for $(\beta, (A_{\epsilon}^{E})^{\prime})$ is obtained by similar 
means, in addition making use of ergodicity. We give a proof in Appendix \ref{app_3}.

Theorems \ref{thm_quantacc} and \ref{coro_regspec} form the core of our method to determine the Lyapunov exponent of extended Harper's model. 
They characterize the $\epsilon$-dependence of the Lyapunov exponent of the complexified cocycles $(\beta,A_{\epsilon}^{E})$ and $(\beta, B_{\epsilon}^{E})$ for energies in the spectrum.

\subsection{Continuity of the LE for singular analytic cocycles} \label{sec_cont}

Before applying the results of the previous section to extended Harper's equation, we present a surprisingly simple proof of the continuous dependence of the LE on the cocycle upon variaton over the analytic category, further illustrating the use of complexified cocycles. 

First recall the following basic, however very useful fact, even valid for {\em{continuous}} $M_2(\mathbb{C})$-cocycles:
\begin{theorem} \label{thm_uscle}
$L: \mathbb{T} \times \mathcal{C}^{0}(\mathbb{T}, M_2(\mathbb{C})) \to \mathbb{R} \cup \{-\infty\}$ is upper-semicontinuous. 
\end{theorem}
\begin{remark}
As an immediate corollary, we have that $L: \mathbb{T} \times \mathcal{C}^{0}(\mathbb{T}, M_2(\mathbb{C})) \to \mathbb{R} \cup \{-\infty\}$ is continuous at every $(\beta,D)$ with $L(\beta, D)= -\infty$. By the same reason, $L: \mathbb{T} \times \mathcal{C}^{0}(\mathbb{T}, SL_2(\mathbb{C})) \to \mathbb{R}_{0+}$ is continuous at every $(\beta,D)$ with $L(\beta, D)= 0$. Here, $\mathbb{R}_{0+}:=\{x \in \mathbb{R}: x \geq 0\}$.
\end{remark}
Because of its usefulness and since the argument is very short, for the reader's convenience, we include a proof of Theorem \ref{thm_uscle} in Appendix \ref{app_uscle}.

Given $\delta>0$, the following subsets of $\mathcal{C}_\delta^\omega(\mathbb{T},M_2(\mathbb{C}))$ will be of interest in the subsequent discussion:
\begin{equation} \label{eq_space1}
\mathcal{A}_\delta^\omega(\mathbb{T},M_2(\mathbb{C}))
:=\{ D \in \mathcal{C}_\delta^\omega(\mathbb{T},M_2(\mathbb{C})): \det D(x) \not \equiv 0 \} ~\mbox{,}
\end{equation}
\begin{equation}
\mathcal{B}_\delta^\omega(\mathbb{T},M_2(\mathbb{C})) := \left\{ D \in \mathcal{C}_\delta^\omega(\mathbb{T},M_2(\mathbb{C})) : \mathcal{N}(\det D; \mathbb{T})=0 \right\} ~\mbox{.}
\end{equation}

We note that $\mathcal{A}_\delta^\omega(\mathbb{T},M_2(\mathbb{C}))$ is open in the Banach space $\mathcal{C}_\delta^\omega(\mathbb{T},M_2(\mathbb{C}))$. As before, it is useful to consider $\mathcal{A}^\omega(\mathbb{T},M_2(\mathbb{C})):=\cup_{\delta>0} \mathcal{A}_\delta^\omega(\mathbb{T},M_2(\mathbb{C}))$ which is open in $\mathcal{C}^\omega(\mathbb{T},M_2(\mathbb{C}))$ relative to the inductive limit topology introduced at the beginning of Sec. \ref{sec_jacobico}. Analogously, one defines $\mathcal{B}^\omega(\mathbb{T},M_2(\mathbb{C}))$, which is open in $\mathcal{A}^\omega(\mathbb{T},M_2(\mathbb{C}))$.

In \cite{CC}, we proved that for a given Diophantine $\beta$, $L(\beta, .)$ is continuous on $\mathcal{A}_\delta^\omega(\mathbb{T},M_2(\mathbb{C}))$. We recall that $\beta$ is called Diophantine if there exists $0 < b(\beta)$ and $1 < r(\beta)< + \infty$ s.t. for all $j \in \mathbb{Z}\setminus \{0\}$
\begin{equation} \label{eq_diophantine}
\abs{\sin(2\pi j \beta)} > \dfrac{b(\beta)}{\abs{j}^{r(\beta)}} ~\mbox{.}
\end{equation}
The Diophantine condition was imposed in order to deal with singularities of the cocycle. Even though it was speculated in \cite{CC} (and in an earlier, related result \cite{A}) that the theorem should hold true for all irrational $\beta$, actual proof did not follow from the method of \cite{CC}. 

In particular, the theorem proven in \cite{CC} does not imply continuity in the frequency $\beta$. For practical purposes, however, continuity upon rational approximation is most desirable since the LE of rational approximates takes a simple form due to periodicity of the cocycle (see (\ref{eq_le_rat})). The main result of this section, Theorem \ref{thm_conti_sing}, also implies continuity in this sense.

Questions of continuity of the LE have been actively studied in the recent past; for a brief survey we refer the reader to the introduction of \cite{CC,DD}. For the present development we will make use of the continuity statement obtained for {\em{non-singular}} cocycles as proven in  \cite{C}; for further use we state this result here: \footnote{In \cite{C}, Theorem \ref{thm_bj} was stated and proven for the Schr\"odinger, $SL(2,\mathbb{R})$ case; strictly speaking, the extension to $SL(2,\mathbb{C})$ follows from \cite{A}. The obvious generalization to non-singular cocycles has been carried out explicitly in Sec. 4 of \cite{DD}.}
\begin{theorem}[\cite{C}] \label{thm_bj}
The Lyapunov exponent $L(\beta+., .): \mathbb{T} \times \mathcal{B}^\omega(\mathbb{T},M_2(\mathbb{C})) \to \mathbb{R}$ is jointly continuous at every {\em{irrational}} $\beta$. 
\end{theorem}

Let us postpone for a moment the proof of Theorem \ref{thm_conti_sing} and first show that continuity in the analytic category is the best one can expect (Theorem \ref{thm_conti_sing_opt}). Recalling (\ref{eq_qa2}), it is enough to proof the following statement:
\begin{prop} \label{prop_contile_top}\begin{enumerate}
\item
The function $I: \mathcal{C}^\infty(\mathbb{T},\mathbb{C}) \to \mathbb{R} \cup \{-\infty\}$, defined by
\begin{equation}
I(c):=\int_\mathbb{T} \log \abs{c(x)} \ud x ~\mbox{,} ~c \in \mathcal{C}^\infty(\mathbb{T},\mathbb{C}) ~\mbox{,}
\end{equation}
is discontinuous at 
\begin{equation}
c(x)=\begin{cases} \mathrm{e}^{-1/\sqrt{\vert\sin(2 \pi x)\vert}} & \mbox{, if} ~x>0 ~\mbox{,} \\ 0 & \mbox{, otherwise.} \end{cases}
\end{equation}
Here, $\mathcal{C}^\infty(\mathbb{T},\mathbb{C})$ is topologized in analogy to (\ref{eq_metricsck}).
\item For $k\in \mathbb{N}_0,$ $I: \mathcal{C}^\omega(\mathbb{T},\mathbb{C}) \to \mathbb{R}$ is discontinuous
at $c(x)=\sin^{2k}(2\pi x)$ in the topology of $\mathcal{C}^k(\mathbb{T},\mathbb{C}).$ 
\end{enumerate}
\end{prop}
\begin{remark} 
\begin{itemize}
\item[(i)] Based on Jensen's formula, we showed in \cite{CC} (see Lemma 2.9 therein) that $I$ is continuous when defined on $\mathcal{C}^\omega(\mathbb{T},\mathbb{C})$ (in inductive limit topology). \label{rem_contlemma}
%\footnote{In fact, the proof of  \ref{prop_contile_top} also shows that $\mathcal{C}^\omega(\mathbb{T},\mathbb{C})$ equipped only with the weaker topology induced by the supremum norm on $\mathbb{T}$ fails to be continuous!})
\item[(ii)] Concavity of the $\log$ and bounded convergence implies that for any $k \in \mathbb{N}_0 \cup \{\infty\}$, $I$ is continuous at any $c$ bounded away from zero. 
%\item[(iii)] The same proof shows discontinuity of $I$ in $\mathcal{C}^k$, if one considers $c(x) = \sin^{2k+1}(2 \pi x)$.
\label{rem_counter}
\end{itemize}
\end{remark}

\begin{proof}
Let $\phi(x)$ be the standard $\mathcal{C}^\infty$ bump-function with  $\mathrm{supp} \phi \in [-2,2]$ and $\phi(x) = 1$ for $x \in [-1,1]$. \begin{comment}i.e.
\begin{eqnarray}
\phi(x) = g(x+2) g(2-x) ~\mbox{,} \\
g(x) = \dfrac{h(x)}{h(x) + h(1-x)} \mbox{,} \\ 
h(x) = \begin{cases} \mathrm{e}^{-1/x^2} & \mbox{, if} ~x>0 ~\mbox{,} \\ 0 & \mbox{, otherwise.} \end{cases}
\end{eqnarray}
\end{comment}
Set $\tilde{\phi}(x):=1-\phi(x)$ and let $a_k= \Vert \phi^{(k)} \Vert_{[-2,2]}.$

Consider the sequence of $\mathcal{C}^\infty$ functions on $\mathbb{T}$ given by
\begin{eqnarray}\label{fn}
- f_n:= & c(x) \left[\phi(4n (x-\frac{1}{2})) + \phi(4n x) + \phi(4n (x-1))\right] \\
& + c(\frac{1}{4n}) \left[\tilde{\phi}(4n (x-\frac{1}{2})) + \tilde{\phi}(4n x) + \tilde{\phi}(4n (x-1))\right]  ~\mbox{.}
\end{eqnarray}

%Direct computation shows $\rho_\infty(f_n,0) \to 0$ as $n\to \infty$. For the reader's convenience, we provide some details in Appendix \ref{app_comp}. 
\begin{claim}
$\rho_\infty(f_n,0) \to 0$ as $n\to \infty$.
\end{claim}
\begin{proof}
Since the remaining terms in (\ref{fn}) can be treated similarly, we will focus on showing that
\begin{equation} \label{eq_wtsderiv}
\rho_\infty(\phi(4n x) c(x),0) \to 0 ~\mbox{, as} ~n\to \infty ~\mbox{.}
\end{equation}

By induction, for $k \in \mathbb{N}$ we have for $x\not= 0$
\begin{eqnarray}
c^{(2k)}(x)=\mathrm{e}^{-\frac{1}{\sqrt{|\sin 2\pi x|}}} |\sin 2\pi x|^{-3k} p_{6k-1}(\sqrt{|\sin 2\pi x|}),\\
c^{(2k+1)}(x)=\mathrm{e}^{-\frac{1}{\sqrt{|\sin 2\pi x|}}} |\sin 2\pi x|^{-3(k+1/2)}p_{6k}(\sqrt{|\sin 2\pi x|})\cos 2\pi x, 
\end{eqnarray}
and $c^{(k)}(0)=0,$ where $p_k(x)$ is a polynomial of degree $k.$
\begin{comment}\begin{equation}
\abs{c^{(k)}(x)} \leq C \mathrm{e}^{-\frac{1}{\sqrt{x}}} x^{-\frac{3 k}{2}} p_{k-1}(\sqrt{x}) ~\mbox{, } 0< \abs{x} \leq 1/4 ~\mbox{,}
\end{equation}
some $C>0$ and a polynomial $p_{k-1}(x)$ of degree $k-1$.
\end{comment}
Thus $\forall d,k \in \mathbb{N}$,
\begin{equation} \label{eq_estimderiv3}
n^d \norm{c^{(k)}(x)}_{[-\frac{1}{2 n}, \frac{1}{2n}]} =: n^d b_k^n \to 0 ~\mbox{, as} ~n \to \infty ~\mbox{.}
\end{equation}

Given $\epsilon>0$, choose $K \in \mathbb{N}$ such that $\sum_{k > K} 2^{-k} < \epsilon$. The claim follows if we can argue that
\begin{equation}
\sum_{k=1}^{K} 2^{-k} \dfrac{\rho_k(\phi(4n x) c(x),0)}{1 + \rho_k(\phi(4n x) c(x),0)} \leq \sum_{k=1}^{K} 2^{-k} \rho_k(\phi(4n x) c(x),0) \to 0 ~\mbox{, } ~n \to \infty ~\mbox{.}
\end{equation}

Indeed, by (\ref{eq_estimderiv3}) for each $k\in \mathbb{N}$ we have
\begin{equation}
\Vert \partial^k \left( c(x) \phi(4n x) \right) \Vert \leq \sum_{p=0}^{k} \binom{k}{p} a_p (4n)^p b^n_{k-p} \to 0 ~\mbox{, } ~n\to \infty ~\mbox{.}
\end{equation}
%Here, $\Vert \phi^{(k)} \Vert_{[-2,2]} =: a_k$, for $k \in \mathbb{N}$.
\end{proof}

Approximate $f_n$ by trigonometric polynomials $p_n$ satisfying
$\|f_n-p_n\|_k < (\frac12)^n$ and set $c_n=c+p_n.$ 
For $n\in \mathbb{N}$, let $J_n:= [0, \frac{1}{4n}) \cup [\frac{1}{2} - \frac{1}{4n}, \frac{1}{2} + \frac{1}{4n}) \cup [1-\frac{1}{4n},1)$. 

We estimate:
\begin{equation} \label{eq_counter1}
I(c_n) \leq -  \log(2) + \int_{\mathbb{T} \setminus J_n} \log \abs{c_n} \ud x   ~\mbox{,}
\end{equation}
for all $n \in \mathbb{N}$. Since the integral on the right hand side of (\ref{eq_counter1}) converges as $n\to\infty$ to $I(c)>-\infty$, we have $I(c_n) \nrightarrow I(c)$.

The second statement follows in the same way with $c(x)=\sin^{2k}(2 \pi x).$ 
%< \frac{1}{4n}$, 

%Take  $M \geq 100$ such that 
%\begin{equation}
%\left\vert \int_{\mathbb{T}\setminus I_M} \log \abs{c_n} \ud x + \log(2) \right\vert < \left(\frac{1}{10}\right)^{100} ~\mbox{.}
%\end{equation}
%Thus, for $n > M$, (\ref{eq_counter1}) implies
%\begin{equation}
%I(c_n) \leq -2\log2 + \left(\frac{1}{10}\right)^{100} + \int_{J_M \setminus J_n} \log\abs{c_n} \ud x ~\mbox{,}
%\end{equation}
%which since $\abs{c_n} < 1$ on $J_M\setminus J_n$, yields $I(c_n) \leq - (3/2) \log2$ for any $n > M$.

\end{proof}

\begin{proof}[Proof of Theorem \ref{thm_conti_sing}]
The key observation, based on the idea of Artur Avila,  is that by convexity of the LE in $\epsilon$, it suffices to show continuity of the LE {\em{away}} from $\mathbb{T}$, i.e. on  $0<\abs{\epsilon}\leq \delta$ for some small $\delta>0$. 

To see that this indeed is sufficient, suppose for a moment the LE was continuous away from $\mathbb{T}$. Then for any $\delta^*$ with $0<\delta^*<\delta$, $L(\beta + r, B_\epsilon) \to L(\beta, A_\epsilon)$ uniformly on $\delta^* \leq \abs{\epsilon} \leq \delta$ as $r \to 0$ and $\norm{B-A}_\delta \to 0$. 
\begin{lemma} \label{fact_conv}
Consider $\mathcal{F}:=\{f:[0,1] \to \mathbb{R}  ~\mbox{convex}\}$ as a closed subset of $(\mathcal{C}([0,1],\mathbb{R}),\Vert . \Vert_{[0,1]})$. For $x\in (0,1)$ denote by $\mathrm{D}_\pm(f)(x)$ the right$(+)$ and left$(-)$ derivative of $f \in \mathcal{F}$, respectively. Given $K \subset (0,1)$ compact, define $T_K^{\pm}(f):= \sup_{x \in K} \abs{\mathrm{D}_\pm(f)(x)}$ for $f \in \mathcal{F}$. Then, $T_K^{\pm}$ is locally bounded.
\end{lemma}
\begin{remark}
This is just a more elaborate version of the following basic fact for a sequence $f_n:(0,1) \to \mathbb{R}$ of convex functions: If $f_n \to f$ pointwise,  then
$D_{-}(f) \leq \liminf_{n \to\infty} D_{-}(f_n) \leq \limsup_{n \to \infty} D_{+}(f_n) \leq D_{+}(f)$ (e.g. \cite{EEE}). We prove Lemma  \ref{fact_conv} in Appendix \ref{app_convlemma}.
\end{remark}

Thus, the acceleration is locally bounded, i.e. there exist $\eta, \kappa, K>0$ such that whenever $\norm{B-A}_\delta < \eta$ and $\abs{r} < \kappa$ we have
\begin{equation}
\sup_{\delta^* \leq \abs{\epsilon} \leq \delta}  \abs{\omega(\beta+r, B;\epsilon)} \leq K ~\mbox{,}
\end{equation}
which since $\omega$ monotonically increases in $\epsilon$ yields
\begin{equation}
\sup_{\abs{\epsilon} \leq \delta}  \abs{\omega(\beta+r, B;\epsilon)} \leq K ~\mbox{.}
\end{equation}

Continuity of the LE then follows upon successive approximation, i.e. 
\begin{eqnarray} 
\left\vert L(\beta + r, B) - L(\beta, A) \right\vert \leq & \left\vert L(\beta + r, B) - L(\beta + r, B_{\delta^{*}}) \right\vert  \nonumber \\ \label{eq_conti1a}
& + \left\vert L(\beta + r, B_{\delta^{*}}) -L(\beta, A_{\delta^{*}}) \right\vert + \left\vert  L(\beta, A_{\delta^{*}}) - L(\beta, A) \right\vert \nonumber \\
\leq & 2 K \delta^{*} + \left\vert L(\beta + r, B_{\delta^{*}}) - L(\beta, A_{\delta^{*}}) \right\vert  ~\mbox{,} \label{eq_conti1}
\end{eqnarray}
for every $\norm{B - A} < \eta$ and $\abs{r} < \kappa$. 

Since $\delta^{*}$ was arbitrarily, Theorem \ref{thm_conti_sing} would follow provided that the second term in (\ref{eq_conti1}) can be made arbitrarily small.

For $\mathcal{A}^\omega(\mathbb{T},M_2(\mathbb{C}))$ this is an immediate consequence of Theorem \ref{thm_bj}, since for $\delta>0$ sufficiently small, $\det A(z) \neq 0$ if $0 < \abs{\im(z)} < \delta$. Thus,
\begin{theorem} \label{thm_contile_v1}
Given $\beta$ irrational, Theorem \ref{thm_conti_sing} holds for the restriction of $L(\beta+.,.)$ to $\mathbb{T} \times \mathcal{A}^\omega(\mathbb{T},M_2(\mathbb{C}))$.
\end{theorem}

What is left is to consider is the case when the determinant vanishes {\em{identically}} on $\mathbb{T}$. From a dynamical point of view this situation is analogous to the case of uniformly hyperbolic cocycles (with obvious abuse of terminology since the determinant is identically zero), where it is known that the LE behaves continuosly. 

That this reasoning may be used to extend above results on continuity to the case of identically vanishing determinant, was also pointed out to us by Artur Avila which we gratefully acknowledge. Following we give the details to this line of argument, which combined with Theorem \ref{thm_contile_v1} completes the proof of Theorem \ref{thm_conti_sing}.

Fix a cocycle $(\beta, A)$ with irrational $\beta$ and $\det A \equiv 0$ on $\mathbb{T}$. Clearly, this automatically implies $\det A \equiv 0$ on $\mathbb{T}_\delta$, the domain of holomorphicity of $A$. 

Consider first the case of  non-trivial $A$ (i.e. $A \not \equiv 0$ on $\mathbb{T}$). Continuity for trivial $A$ is dealt with in Lemma \ref{lem_contleconst}. 

Without loss of generality we may then assume that the first row of $A$ does not vanish identically on $\mathbb{T}$. Thus, there exists $0<\delta^*<\delta$ such that for all $z \in \mathbb{T}_{\delta^*} \setminus \mathbb{T}$, $(A_{11}(z), A_{12}(z)) \neq 0$, whence $\dim \mathrm{Ker} A(z) = \dim \mathrm{Ran} A(z) = 1$ on $\mathbb{T}_{\delta^*} \setminus \mathbb{T}$.

Letting, 
\begin{equation}
u(z):=\begin{pmatrix}  -A_{12}(z) \\ A_{11}(z)  \end{pmatrix} ~\mbox{,} ~v(z) := A(z) \begin{pmatrix} A_{11}(z) \\ A_{12}(z) \end{pmatrix} ~\mbox{,}
\end{equation}
we thus obtain analytic functions defined on $\mathbb{T}_{\delta^*}\setminus \mathbb{T}$ with $\mathrm{Span}\{u(z)\} =  \mathrm{Ker} A(z)$ and $\mathrm{Span} \{v(z)\} = \mathrm{Ran} A(z)$.

% Without loss of generality, we may assume $A_{11}$ does not vanish identically, i.e. $A$ can be written as
%\begin{equation}
%A(z) = \begin{pmatrix} a_1(z) & a_2(z) \\ \alpha(z) a_1(z) & \alpha(z) a_2(z)  \end{pmatrix} ~\mbox{,}
%\end{equation}
%for some functions $a_1, a_2, \alpha \in \mathcal{C}^\omega_\delta(\mathbb{T},\mathbb{C})$, where there exists $0<\delta^*<\delta$ such that for all $z \in \mathbb{T}_{\delta^*} \setminus \mathbb{T}$  $a_1(z) \neq 0$. Notice that $\alpha(z)$ may be identically zero. 
%
%In particular, $\dim \mathrm{Ker} A(z) = \dim \mathrm{Ran} A(z) = 1$ for all $z \in \mathbb{T}_{\delta^*} \setminus \mathbb{T}$. Obviously, on $\mathbb{T}_{\delta^*} \setminus \mathbb{T}$, both $\mathrm{Ker}  A(z)$ as well as $\mathrm{Ran} A(z)$ depend analytically on $z$, in fact:
%\begin{equation} \label{eq_kerran}
%\mathrm{Ker}  A(z) = \mathrm{Span} \left\{ \begin{pmatrix} -a_2(z) \\ a_1(z)  \end{pmatrix}   \right\} ~\mbox{,} ~ \mathrm{Ran} A(z) =  \mathrm{Span} \left\{  \begin{pmatrix} 1 \\ \alpha(z) \end{pmatrix}  \right\} ~\mbox{.}
%\end{equation}
%
%Equation (\ref{eq_kerran}) explicitly shows that $\mathrm{Ker}  A(z)$ and $\mathrm{Ran} A(z)$ are transversal if and only if $a_1(x) + a_2(x) \alpha(x-\beta) \not \equiv 0$ on $\mathbb{T}$. 

For $z \in \mathbb{T}_{\delta^*} \setminus \mathbb{T}$, define $C(z):= \left(v(z-\beta) \vert u(z) \right)$. We distinguish the following two cases:
\begin{description}
\item[Case 1] $\det C(z) \not \equiv 0$ on $\mathbb{T}_{\delta^*} \setminus \mathbb{T}$. Then, making $\delta^*>0$ sufficiently small, $\mathrm{Ker}  A(z)$ and $\mathrm{Ran} A(z-\beta)$ are transversal i.e. $C(z) \in GL(2,\mathbb{C})$ for all $0<\abs{\im{z}} < \delta^*$.
%, i.e. for $\delta^*$ sufficiently small, $a_1(z) + a_2(z) \alpha(z-\beta) \neq 0$ on $\mathbb{T}_{\delta^*} \setminus \mathbb{T}$

For $z \in \mathbb{T}_{\delta^*} \setminus \mathbb{T}$ we thus obtain an analytic conjugacy of the cocycle $(\beta, A)$ given by
\begin{equation} \label{eq_conti_singidvanish_2}
C(z+\beta)^{-1} A(z) C(z) = \begin{pmatrix} c(z) & 0 \\ 0 & 0  \end{pmatrix} =:D(z) ~\mbox{.}
\end{equation}
Here, $c(z)$
% = a_1(z) + a_2(z) \alpha(z-\beta)
 %= \det C_1(z)$
 is holomorphic and non-zero on $\mathbb{T}_{\delta^*} \setminus \mathbb{T}$. 

We thus obtain,
\begin{equation} \label{eq_conti_singidvanish_1}
L(\beta, A_\epsilon) = L(\beta, D_{\epsilon}) = \int_{\mathbb{T}} \log\abs{c(x+i\epsilon)} \ud x ~\mbox{,}
\end{equation}
for $0 < \abs{\epsilon} < \delta^*$.
From above construction, it is clear that $c(z)$ depends continuously on both $A$ and $\beta$. 

\item[Case 2] If $\det C(z) \equiv 0$ on $\mathbb{T}_{\delta^*} \setminus \mathbb{T}$, $\mathrm{Ker}  A(z)$ and $\mathrm{Ran} A(z-\beta)$ coincide on $\mathbb{T}_{\delta^*} \setminus \mathbb{T}$. In this case $A(z+\beta)A(z)\equiv 0$ on $\mathbb{T}_{\delta^*}$. 
%\begin{comment} can  be conjugated to a nilpotent cocycle: letting
%\begin{equation}
%C_2(z) := \begin{pmatrix} 1 & 0 \\ -\alpha(z-\beta) & 1 \end{pmatrix} ~\mbox{,}
%\end{equation}
%for $z \in \mathbb{T}_{\delta^*}\setminus \mathbb{T}$, we compute: 
%\begin{equation} \label{eq_conti_singidvanish_2}
%C_2(z+\beta) A(z) C_2(z)^{-1} = \begin{pmatrix} 0 & a_2(z) \\ 0 & 0  \end{pmatrix} ~\mbox{.}
%\end{equation}
%\end{comment}
Thus $L(\beta,A_\epsilon)=-\infty$, $\abs{\epsilon} \leq \delta^*$ which directly leads to continuity by Theorem \ref{thm_uscle}.
%Notice that $\abs{a_1(z)} = \abs{a_2(z) \alpha(z-\beta)}$ together with our assumptions on $a_1(z)$ and $\alpha(z)$ imply that $a_2(z)$ is bounded from zero on  $\mathbb{T}_{\delta^*}\setminus \mathbb{T}$.
\end{description}

Since  in Case 1 ($c(z) \neq 0$ on $0<\abs{\im z} \leq \delta^*$),
\begin{equation}
L(\beta, A_\epsilon) = L(\beta, \begin{pmatrix} c(x+i\epsilon) & 0 \\ 0 & 0 \end{pmatrix}) = \int_\mathbb{T} \log \abs{c(x+i\epsilon)} \ud x + L(\beta, \begin{pmatrix}  1 & 0 \\ 0 & 0 \end{pmatrix}) ~\mbox{,}
\end{equation}
we conclude that for $0 < \abs{\epsilon} \leq \delta^*$, stability of the LE under {\em{any}} analytic perturbation of $(\beta, A)$ reduces to showing continuity at the {\em{constant}} cocycle $A =
% \left\{ 
\begin{pmatrix} 1 & 0 \\ 0 & 0 \end{pmatrix}
%, \begin{pmatrix} 0 & 1 \\ 0 & 0 \end{pmatrix} 
%\right\}
$. More generally, the following is true:

\begin{lemma} \label{lem_contleconst}
%Let $\alpha \in \{0, \omega, \infty\} \cup \mathbb{N}$. 
The Lyapunov exponent $L( ., .): \mathbb{T} \times \mathcal{C}^0(\mathbb{T}, M_2(\mathbb{C})) \to \mathbb{R}\cup\{-\infty\}$ is (jointly) continuous at any constant cocycle. 
%Here, we define $L(\beta, 0):=-\infty$ in accordance with Remark \ref{rem_defletriv} (iii). 
\end{lemma}
\begin{remark}
Note that similarly to the continuity of the LE over the uniformly hyperbolic cocycles, the frequency does not have to be irrational.
\end{remark}
\begin{proof}
Let $(\beta, A)$ be a constant cocycle. We distinguish three cases depending on the relative magnitude of the eigenvalues of $\lambda_1, \lambda_2$ (counting multiplicities) of $A$.

The case $\abs{\lambda_1} \neq \abs{\lambda_2}$, was examined (in greater generality) by Ruelle in \cite{FFF}. We mention that Ruelle's formulation only states continuity (in fact even real-analyticity) in the matrix-valued function, however it is clear from his proof that the same strategy shows joint continuity in $(\beta, B)$.

If $\lambda_1 = \lambda_2 = 0$, $L(\beta, A) = -\infty$ whence continuity follows trivially from upper-semicontinuity, Theorem \ref{thm_uscle}. 

Finally, if $\abs{\lambda_1} = \abs{\lambda_2} \neq 0$, $\det A \neq 0$, the remark to Theorem \ref{thm_uscle} immediately implies continuity of the LE for the ``renormalized'' cocycles $(\beta, B/\sqrt{\abs{\det B}})$ at $(\beta, A/\sqrt{\abs{\det(A)}})$. This, in turn implies the claim by Remark \ref{rem_counter} (ii).
\end{proof}
\end{proof}

\section{Extended Harper's model} \label{sec_ehm}

To simplify notation, throughout the paper we shall suppress some of the dependencies of $H_{\lambda,\theta;\beta}$ if the context permits. Following we summarize some statements necessary to apply above results for general Jacobi cocycles to extended Harper's model.

First we mention:
\begin{obs} \label{obs_cfun}
Letting $z=x+i \epsilon$, $c(z)$ has at most two zeros on the strip $0 \leq x < 1$; the zeros occur if at least one of the following conditions
is satisfied: 
\begin{eqnarray}
\lambda_{1} \mathrm{e}^{-2 \pi \epsilon} & = & \lambda_{3} \mathrm{e}^{2 \pi \epsilon} ~\mbox{,}\\
\lambda_{1} \mathrm{e}^{-2 \pi \epsilon} + \lambda_{3} \mathrm{e}^{2 \pi \epsilon} & = & \pm \lambda_{2} ~\mbox{.}
\end{eqnarray}
Necessary conditions for real roots are $\lambda_{1} = \lambda_{3}$ or $\lambda_{1} + \lambda_{3} = \lambda_{2}$.
Moreover, for $\lambda_{1} = \lambda_{3}$, $c(z)$ has real roots if and only if $2\lambda_{3} \geq \lambda_{2}$.  If  $\lambda_1 \neq \lambda_3$ real roots  occur if and only if $\lambda_1 + \lambda_3 = \lambda_2$.
\end{obs}

Application of Jensen's formula yields  (see \cite{E} for $\epsilon =0$),
\begin{equation} \label{eq_integral}
I_{\epsilon}(\lambda) = \begin{cases} \log \lambda_{3} + 2 \pi \epsilon & \mbox{if} ~\lambda_{3} \mathrm{e}^{2 \pi \epsilon} \geq  \lambda_{1} \mathrm{e}^{- 2 \pi \epsilon} \geq 0 \\ & ~\mbox{and} ~\lambda_{1} \mathrm{e}^{-2 \pi \epsilon} + \lambda_{3} \mathrm{e}^{2 \pi \epsilon} \geq \lambda_{2} \geq 0 ~\mbox{,} \\
\log{\lambda_{1}} - 2 \pi \epsilon & \mbox{if} ~\lambda_{1}  \mathrm{e}^{-2 \pi \epsilon} \geq \lambda_{3} \mathrm{e}^{2 \pi \epsilon} \geq 0  \\& ~\mbox{and} ~\lambda_{1} \mathrm{e}^{ - 2 \pi \epsilon} + \lambda_{3} \mathrm{e}^{2 \pi \epsilon} \geq \lambda_{2} \geq 0 ~\mbox{,} \\
\log \left \vert \dfrac{2\lambda_{1}\lambda_{3}}{-\lambda_{2}+\sqrt{\lambda_{2}^2 - 4\lambda_{1}\lambda_{3}}} \right \vert & \mbox{if} ~\lambda_{1} \mathrm{e}^{- 2 \pi \epsilon} +\lambda_{3} \mathrm{e}^{2 \pi \epsilon} \leq \lambda_{2} ~\mbox{and} ~ \lambda_{1},\lambda_{3} \neq 0 ~\mbox{,}\\
\log \lambda_{2} & \mbox{if} ~\lambda_{1} \mathrm{e}^{-2 \pi \epsilon} + \lambda_{3} \mathrm{e}^{2 \pi \epsilon} \leq \lambda_{2} ~\mbox{,} ~ \lambda_{1} ~\mbox{or}  ~\lambda_{3} = 0 ~\mbox{,} \end{cases}
\end{equation}
where we set $I(c_\lambda) =:I(\lambda)$ in (\ref{eq_Ifun}).

We note that (\ref{eq_integral}) explicitly shows that $I_{\epsilon}(\lambda)$ is jointly continuous in $\epsilon$ and $\lambda$.

Application of Theorem \ref{thm_conti_sing} to the model under consideration results in 
\begin{theorem} \label{thm_lana}
For fixed irrational $\beta$, the Lyapunov exponents $L(\beta+r;B_{\epsilon}^{E},\lambda)$ and $L(\beta+r;A_{\epsilon}^{E},\lambda)$ are jointly
continuous in $\epsilon$, $ \lambda$, $E$ and $r$. 
\end{theorem}

We emphasize that Theorem \ref{thm_quantacc} and \ref{thm_unifhyp} require analytic cocycles with determinant bounded away from zero. As discussed in Observation \ref{obs_cfun},
for any $\lambda$ with $\lambda_{1} \neq \lambda_{3}$ or $\lambda_{2} \neq \lambda_{1}+\lambda_{3}$, $c(z)$ is non-zero on the strips
$\{\abs{\im{z}} < \epsilon_{1}\}$ and  $\{\abs{\im{z}} > \epsilon_{2}\}$, for some $0 < \epsilon_{1} < \epsilon_{2}$. 

Thus the hypotheses of Theorem \ref{thm_quantacc} and \ref{thm_unifhyp} apply to $(\beta, A_{\epsilon}^{E})$ for small and large $\epsilon$ (in magnitude) when excluding 
$\lambda_{1} = \lambda_{3}$ and $\lambda_{2} = \lambda_{1}+\lambda_{3}$. Let
\begin{equation}
\mathcal{R}:=\left\{\lambda:  \lambda_{1}\neq\lambda_{3}, ~\lambda_{2}\neq\lambda_{1}+\lambda_{3} \right\} ~\mbox{.}
\end{equation}
Notice that upon exclusion of the two planes described by $\mathcal{R}^{c}$, we do not loose any information; the derived statements about $L(\beta, B^{E})$ will extend to all values of $\lambda$ using Theorem \ref{thm_lana} as well as continuity of the spectrum in the Hausdorff metric. 

\section{Asymptotic analysis} \label{sec_asympt}

In this section we aim to obtain an expression for $L(\beta, A_{\epsilon}^{E})$ valid for large $\abs{\epsilon}$ and any choice of $\lambda$. We mention that this strategy can be employed generally to Jacobi cocyles consisting of trigonometric polynomials.

The basic idea is to reduce the non-trivial problem of computing the LE of a given non-constant cocyle to an ``almost constant'' cocycle by taking $\abs{\epsilon} \to \infty$. By Theorem \ref{thm_conti_sing} the LE is stable under small analytic perturbations, whence the task boils down to consideration of a {\em{constant}} cocyle whose LE can be computed trivially.  Finally, quantization of acceleration, convexity and the Theorems \ref{thm_quantacc} and \ref{coro_regspec} allow us extrapolate to $\epsilon=0$. 

Our strategy to determine the Lyapunov exponent was motivated by considerations of Avila's in \cite{B}. There, complexification is employed to prove the 
Aubry-Andr\`e formula (see Appendix A in \cite{B}). 

Since the present model specializes to the almost Mathieu equation when 
$\lambda_{1}=\lambda_{3}=0$,
as a warmup, we present a proof of the Aubrey-Andr\`e formula alternative to Avila's in that it replaces a geometric argument with one along the line of what is to come.

Setting $\lambda_{1}=\lambda_{3}=0$, we obtain the transfer matrix of the almost Mathieu equation
\begin{equation} \label{eq_app4_1}
B^{E}(x) = \begin{pmatrix} E - 2 \mu \cos(2 \pi x)  &  -1 \\  1 & 0 \end{pmatrix} ~\mbox{.}
\end{equation}
To simplify notation we put $\mu = \lambda_{2}^{-1}$. Here $A^E=B^E.$ 

Complexifying (\ref{eq_app4_1}), as $\epsilon \to +\infty$
\begin{equation} \label{d2} 
 B^{E}(x) = \mu \mathrm{e}^{-2\pi i x} \mathrm{e}^{2 \pi \epsilon} \begin{pmatrix} -1 + o(1) &  o(1) \\  o(1) & 0 \end{pmatrix}
              =: \mu \mathrm{e}^{-2\pi i x} \mathrm{e}^{2 \pi
                \epsilon} D_{\epsilon}(x)\end{equation}
uniformly in $x\in\mathbb{T}$. Here, the the $o(1)$-terms are perturbations by trigonometric polynomials with only non-negative harmonics. We mention that for $SL(2,\mathbb{R})$ cocycles the reflection principle implies that $L(\beta, B_{\epsilon}^{E})$ is an even function in $\epsilon$.

Applying Theorem \ref{thm_conti_sing} to (\ref{d2}) we conclude,
\begin{equation} \label{eq_app4_3}
L(\beta, B_{\epsilon}^{E}) = \log{\abs{\mu}} + 2 \pi \abs{\epsilon}, ~ \abs{\epsilon} > \epsilon^{\prime} ~\mbox{,}
\end{equation}
which gives the desired asymptotic formula for the complexified Lyapunov exponent of the almost Mathieu equation.

\begin{remark}
\begin{itemize}
Making use of Theorem \ref{thm_quantacc} and \ref{coro_regspec}, the asymptotic relation (\ref{eq_app4_3}) can easily be extrapolated to $\epsilon = 0$ producing
the complex Aubry-Andr\'e formula
\begin{equation} \label{eq_AubryAndre}
L(\beta, B_{\epsilon}^{E}) \geq \max\{0, \log{\abs{\mu}} + 2 \pi \abs{\epsilon}\} ~\mbox{,}
\end{equation}
with equality holding for $E \in \Sigma$. For details we refer to \cite{B}, Appendix A.
\end{itemize}
\end{remark}
 Following we shall focus on the case when at least one of 
$\lambda_{1}$, $\lambda_{3}$ is nonzero.

\begin{prop} \label{prop_asy}
For at least one of $\lambda_{1}$, $\lambda_{3}$ nonzero,
\begin{equation}
L(\beta, A_{\epsilon}^{E}) = \log\left\vert \dfrac{1+\sqrt{1-4 \lambda_{1} \lambda_{3}} }{2}\right\vert + 2 \pi \abs{\epsilon} ~\mbox{,} ~\abs{\epsilon} > \epsilon_{0} ~\mbox{,}
\end{equation}
some $\epsilon_{0} > 0$.
\end{prop}
\begin{proof}
\begin{description}
\item[Case 1 - both $\lambda_{1}$ and $\lambda_{3}$ are nonzero]
Then, uniformly for $x \in \mathbb{T}$
\begin{eqnarray} \label{eq_asy1a}
A^{E}(x+i\epsilon) = \mathrm{e}^{2 \pi \epsilon} \mathrm{e}^{-2\pi i x} M_{+,\epsilon}(\lambda_{1}, \lambda_{3}) ~\mbox{,} \nonumber \\
M_{+,\epsilon}(\lambda_{1}, \lambda_{3}):=\begin{pmatrix} -1 + o(1)  &  -\lambda_{1} \mathrm{e}^{\pi i \beta} + o(1) \\ \lambda_{3} \mathrm{e}^{-\pi i \beta} + o(1) & 0   \end{pmatrix} ~\mbox{,} \label{eq_defm}
\end{eqnarray}
as $\epsilon \to +\infty$. In particular, for the LE this implies,
\begin{equation} \label{eq_asy1a1}
L(\beta, A^E) = 2 \pi \epsilon + L(\beta, M_{+,\epsilon}) ~\mbox{.}
\end{equation}

Note that $(\beta, M_{+,\epsilon})$ is an almost constant analytic cocycle,
\begin{equation}
M_{+,\epsilon}(x) = M^{(0)} + M_{+,\epsilon}^{(1)}(x), ~ M^{(0)} := \begin{pmatrix} -1 & -\lambda_1 \mathrm{e}^{i \pi  \beta} \\ \lambda_3 \mathrm{e}^{- i \pi \beta} & 0 \end{pmatrix} ~\mbox{,}
\end{equation}
as $\epsilon \to +\infty$.
Thus, by continuity of the LE, 
\begin{equation}
L(\beta, M_{+,\epsilon}) = L(\beta, M^{(0)}) + o(1) ~\mbox{,}
\end{equation}
as $\epsilon \to +\infty$.

In summary, (\ref{eq_asy1a1}) and quantization of acceleration yield
\begin{equation} \label{eq_asy1b}
L(\beta,A_{\epsilon}^{E}) = 2 \pi \epsilon  + L(\beta, M^{(0)}) ~\mbox{, all $\epsilon > 0$ sufficiently large.}
\end{equation}
Since the matrix $M$ in (\ref{eq_defm}) is independent of $x$, the problem is reduced to computing the Lyapunov exponent
of the {\em{constant}} cocycle $(\beta,M^{(0)})$. 
Solving the eigenvalue problem for $M^{(0)}$, we obtain as claimed
\begin{equation} \label{eq_Mmatrix}
L(\beta,M^{(0)})=\log\left\vert \dfrac{1+\sqrt{1-4 \lambda_{1} \lambda_{3}} }{2}\right\vert ~\mbox{.}
\end{equation}

Using above line of argumentation, the statement about $L(\beta,A_{\epsilon}^{E})$ for $\epsilon$ large and negative readily follows since then
\begin{eqnarray} 
A^{E}(x+i\epsilon) = \mathrm{e}^{-2 \pi \epsilon} \mathrm{e}^{2\pi i x} M_{-,\epsilon} ~\mbox{,} \\
M_{-,\epsilon} = \begin{pmatrix} -1+ o(1)  &  -\lambda_{3} \mathrm{e}^{-\pi i \beta} + o(1) \\ \lambda_{1} \mathrm{e}^{\pi i \beta} + o(1) & 0   \end{pmatrix} ~\mbox{,} 
\end{eqnarray}
and $L(\beta,M^{(0)}) = L(\beta, M_{-}^{(0)})$, where
\begin{equation}
M_{-}^{(0)} = \begin{pmatrix} -1 & -\lambda_{3} \mathrm{e}^{-\pi i \beta} \\ \lambda_{1} \mathrm{e}^{\pi i \beta} & 0   \end{pmatrix} ~\mbox{.}
\end{equation}

\item[Case 2 - one of $\lambda_{1},\lambda_{3} = 0$]
Without loss, we assume $\lambda_{3} = 0$.

Taking $\epsilon  \to \pm \infty$ we obtain uniformly for $x \in \mathbb{T}$,
\begin{eqnarray}
& A^E(x+i\epsilon) = \mathrm{e}^{2 \pi i \epsilon} \mathrm{e}^{- 2 \pi i x} \begin{pmatrix} - 1 + o(1) &  -\lambda_{1} \mathrm{e}^{\pi i \beta} + o(1) \\ o(1) & 0  \end{pmatrix} ~\mbox{,} ~\epsilon \to +\infty ~\mbox{,} \nonumber \\
& A^E(x+i\epsilon) = \mathrm{e}^{- 2 \pi \epsilon} \mathrm{e}^{2 \pi i x} \begin{pmatrix} -1 + o(1) &  o(1) \\ \lambda_{1} \mathrm{e}^{\pi i \beta} + o(1) & 0  \end{pmatrix} ~\mbox{,} ~\epsilon \to -\infty ~\mbox{.}
\end{eqnarray}

The case when $\epsilon \to -\infty$ being similar, we focus on $\epsilon \to +\infty$. To this end write
\begin{equation} \label{eqn_defR1}
A^E(x+i \epsilon) =: \mathrm{e}^{2 \pi \epsilon} \mathrm{e}^{-2 \pi i x} R_{\epsilon}(x) ~\mbox{.}
\end{equation}
Applying Theorem \ref{thm_conti_sing}, we conclude $L(\beta, R_\epsilon) = 0$ for $\epsilon$ sufficiently large and positive. We note that $\epsilon \to +\infty$ in $R_\epsilon$, yields a constant cocycle with identically vanishing determinant. 

Making use of (\ref{eqn_defR1}), by quantization of acceleration one concludes
\begin{equation}
L(\beta, A_\epsilon^E) = 2 \pi \epsilon ~\mbox{, if} ~\epsilon>0  ~\mbox{sufficiently large.}
\end{equation}

\end{description}
\end{proof}
We amend that for any choice of $\lambda$ (i.e. in both of the above cases) the Lyapunov exponent of  $(\beta, A_\epsilon^{E})$ can asymptotically be written in terms of the matrix $M$,
\begin{equation}
L(\beta, A_\epsilon^E) = L(\beta, M) + 2 \pi \abs{\epsilon} ~\mbox{.}
\end{equation}
For later purposes, we note that the matrix $M$ does not depend on $\lambda_{2}$. 

\section{Extrapolation} \label{sec_whatwecansay}
In Proposition \ref{prop_asy} we proved an asymptotic formula for $L(\beta, A_{\epsilon}^{E})$ valid for large $\abs{\epsilon}$. Using convexity of $L(\beta, A_{\epsilon}^{E})$ in $\epsilon$, quantization of acceleration and (\ref{eq_relAB}), we obtain 
the following lower bound for  $L(\beta, B^{E})$,
\begin{equation} \label{eq_lowerbdLE}
L(\beta, B^{E}) \geq \max \{L(\beta, M) - I(\lambda), 0\} ~\mbox{.}
\end{equation} 
Here we also made use of $L(\beta, B^{E}) \geq 0$. We emphasize that in (\ref{eq_lowerbdLE}) we do not require $E$ to be in the spectrum.

For $E \in \Sigma$, Theorem \ref{coro_regspec} will allow us to improve on (\ref{eq_lowerbdLE}) and actually extrapolate $L(\beta, B_{\epsilon}^{E})$ to $\epsilon=0$. 
Returning to Proposition \ref{prop_asy}, the asymptotic expression for $L(\beta, A_{\epsilon}^{E})$ implies that for $\abs{\epsilon} > \epsilon_{0}$ the acceleration 
$\omega(\beta,A^{E};\epsilon)$ is $\pm 1$. 

Let $0<\epsilon^\prime_{1} <\epsilon_{0}$ such that the Theorems \ref{thm_quantacc} and \ref{coro_regspec} apply on $[-\epsilon^\prime_{1},\epsilon^\prime_{1}]$. 
Quantization of acceleration and convexity of $L(\beta, A_{\epsilon}^{E})$  require $\omega(\beta,A_\epsilon^{E})$
to increase on $[-\epsilon^\prime_{1},\epsilon^\prime_{1}]$ with possible values in $\{-1,0,1\}$. 

Moreover, for $E \in \Sigma$ with $L(\beta, B^{E})>0$, by 
Theorem \ref{coro_regspec} the acceleration is forced to jump when passing through $\epsilon=0$. This amounts to three possible situations, a jump 
in $\omega(\beta,A_{\epsilon}^{E})$ from $0 \to 1$, $-1 \to 0$, or $-1 \to 1$. Continuity of $L(\beta, A_{\epsilon}^{E})$ in $\epsilon$, 
rules out the former two cases resulting in $L(\beta, A_{\epsilon}^{E}) =  L(\beta, M) + 2 \pi \abs{\epsilon}$. In particular, (\ref{eq_lowerbdLE}) holds as {\em{equality}} for such $E$.

If, on the other hand, $L(\beta, B_\epsilon^E)=0$, (\ref{eq_lowerbdLE}) yields $L(\beta, A^E) = I(\lambda)$. Again convexity in
$\epsilon$ of $L(\beta, A_\epsilon^E)$ eliminates all possibilities except for
\begin{equation} \label{eq_complexLE_A}
L(\beta, A_\epsilon^E) = \max\left\{I(\lambda), L(\beta,M) + 2 \pi \abs{\epsilon} \right\} ~\mbox{.}
\end{equation} 
Since $L(\beta, M) > I(\lambda)$ whenever $E \in \Sigma$ with $L(\beta, B_\epsilon^E) > 0$, we conclude that the formula in (\ref{eq_complexLE_A}) is in fact valid for all $\lambda$.

We thus obtain the following expression for the Lyapunov exponent of the complexified extended Harper's model, 
\begin{theorem} \label{prop_LE}
For irrational $\beta$,
\begin{eqnarray} 
L(\beta, B_\epsilon^{E}) =  L(\beta, A_\epsilon^E)  - I_\epsilon(\lambda)  ~\mbox{,}  \label{eq_extrap1} \\
L(\beta, A_\epsilon^E) \geq \max\left\{I(\lambda), L(\beta,M) + 2 \pi \abs{\epsilon} \right\} ~\mbox{,} \label{eq_extrap2}
\end{eqnarray}
where equality holds for $E \in \Sigma$. $L(\beta,M)$ and $I_\epsilon(\lambda)$ are given in (\ref{eq_Mmatrix}) and (\ref{eq_integral}), respectively.
\end{theorem}

Finally,  taking $\epsilon = 0$ in Theorem \ref{prop_LE}, Theorem \ref{thm_mainresult} is obtained using the following fact, verified by a straightforward computation.
\begin{obs} \label{thm_anadelta}
Denote $\Delta := L(\beta, M) - I(\lambda)$. Then,
\begin{equation}
\Delta = \begin{cases}\log \left \vert  \dfrac{1+\sqrt{1 - 4\lambda_{1} \lambda_{3}}}{2 \lambda_{1}}\right \vert & \mbox{, if} ~\lambda_{1} \geq \lambda_{3}, ~\lambda_{2} \leq \lambda_{3} + \lambda_{1} ~\mbox{,}  \\
\log \left\vert  \dfrac{1+\sqrt{1 - 4\lambda_{1} \lambda_{3}}}{2 \lambda_{3}}\right\vert & \mbox{, if} ~\lambda_{3} \geq \lambda_{1}, ~\lambda_{2} \leq \lambda_{3} + \lambda_{1} ~\mbox{,}  \\
\log \left\vert  \dfrac{1+\sqrt{1 - 4\lambda_{1} \lambda_{3}}}{\lambda_{2} + \sqrt{\lambda_{2}^{2} - 4 \lambda_{1} \lambda_{3}}}\right\vert & ~\mbox{, if} ~\lambda_{2} \geq \lambda_{3} + \lambda_{1} ~\mbox{.} 
\end{cases}
\end{equation}
$\Delta > 0$ is positive in the interior of region I and $\Delta \leq 0$ for $\lambda \in II \cup III$. 

Moreover,
\begin{itemize}
\item[(i)] For $\lambda_1 \neq \lambda_3$, $\Delta=0$ if and only if $\lambda$ lies on the line segments $\{\lambda_1 + \lambda_3 = 1, \lambda_2 \leq 1\}$ or
$\{\lambda_1 + \lambda_3 \leq 1, \lambda_2 = 1\}$.
\item[(ii)] For  $\lambda_1 = \lambda_3$, $\Delta=0$ if and only if $\lambda \in III$ or $\{2 \lambda_1 \leq 1, \lambda_2 = 1\}$.
\end{itemize}
\end{obs}

As a corollary to Theorem  \ref{prop_LE} we obtain a necessary and sufficient condition for regularity of the analytic cocycle $(\beta, A_\epsilon^E)$. By (\ref{eq_complexLE_A}), $(\beta, A_\epsilon^E)$ is regular if and only if $\Delta < 0$. 

For Schr\"odinger cocycles, non-regular behavior for zero Lyapunov
exponent proved to be interesting. In \cite{B}, Avila called such
cocycles {\em{critical}} and analyzed (lack of) their appearance for generic almost periodic one frequency Schr\"odinger operators \cite{B, O}. 

Motivation for these considerations was the almost Mathieu operator. There, the critical point corresponds to the square lattice, separating {\em{sub-critical}} (zero Lyapunov exponent and regular) from {\em{super-critical}} (positive Lyapunov exponent) behavior (see Eq.  (\ref{eq_AubryAndre})). 

Implications for spectral theory for general analytic Schr\"odinger cocyles are conjectured (almost reducibility conjecture, see Sec. \ref{sec_concl}).

In analogy to the Schr\"odinger case we make following defintion:
\begin{definition}
Consider a Jacobi matrix of the form (\ref{eq_hamiltonian}) with functions $c,v$ analytic and $c \not \equiv 0$. For given $\beta$, we call the associated analytic cocycle $(\beta, A_\epsilon^E)$ {\em{critical}} if the Lyapunov exponent $L(\beta, B^E) = 0$ and $(\beta, A_\epsilon^E)$ shows non regular behavior.
\end{definition}
For extended  Harper's equation Theorem  \ref{prop_LE} and Observation \ref{thm_anadelta} yield 
\begin{coro} \label{coro_regularity}
For extended Harper's equation with irrational frequency $\beta,$ critical behavior occurs
\begin{itemize}
\item[(i)] for $\lambda_1 \neq \lambda_3$ on the line segments $\{\lambda_1 + \lambda_3 = 1, \lambda_2 \leq 1\}$ and $\{\lambda_1 + \lambda_3 \leq 1, \lambda_2 = 1\}$,
\item[(ii)]  for  $\lambda_1 = \lambda_3$, in region III and along $\{2 \lambda_1 \leq 1, \lambda_2 = 1\}$.
\end{itemize}
\end{coro}

\section{An alternative proof of Theorem \ref{thm_mainresult}} \label{sec_alternative}

As pointed out in the end of Sec. \ref{sec_asympt}, asymptotically $L(\beta, A_\epsilon^E)$ is independent of $\lambda_2$. This interesting fact inspires an alternative proof for Theorem \ref{thm_mainresult} given here. First note that above mentioned symmetry is already implied by (\ref{eq_asy1a}) and (\ref{eq_asy1b}) without any further computations
\begin{footnote} {Using Theorem \ref{thm_lana} and continuity of the spectrum in the Hausdorff metric it suffices to establish Theorem \ref{thm_mainresult} for $\lambda_1 \mbox{,} ~\lambda_3$ both non-zero.} \end{footnote}.

The arguments of Sec. \ref{sec_whatwecansay} leading to Theorem \ref{prop_LE} are independent of the remainder of Sec. 
\ref{sec_asympt}; this yields a formal expression for the Lyapunov exponent given in (\ref{eq_extrap1}) and (\ref{eq_extrap2}). 

We will focus here on the most interesting region III. As we will argue, Theorem \ref{thm_mainresult} rests on the following key Lemma which shows that above mentioned asymptotic independence of $L(\beta, B_\epsilon^E)$ on $\lambda_2$ persists in the limit $\epsilon \to 0$.
\begin{lemma} \label{thm_keylemma}
Let $\beta$ be irrational. Then, on the spectrum the Lyapunov exponent in region III is {\em{independent}} of the coupling $\lambda_{2}$.
\end{lemma}
\begin{proof}
First note that for $\lambda_{1} + \lambda_{3} \geq 1$, 
(\ref{eq_integral}) implies $I(\lambda) = \max \left\{\log \abs{\lambda_{1}}, \log \abs{\lambda_{3}} \right\}$. Since the matrix $M$ is
independent of $\lambda_{2}$, within region III Proposition \ref{prop_LE} implies that for fixed $\lambda_{1}$ and $\lambda_{3}$, $L(\beta, B^{E})$ is constant w.r.t. $\lambda_{2}$.  In particular, in entire region III, $L(\beta, B^{E})$ is determined by its value on the plane $\{\lambda_{1} + \lambda_{3} = \lambda_{2} \geq 1\}$. 
\end{proof}

To see that Lemma \ref{thm_keylemma} already implies Theorem \ref{thm_mainresult} notice the following: As an immediate consequence of Lemma \ref{thm_keylemma} we obtain constancy of $L(\beta, B^E)$ on the spectrum along lines where $\lambda_{3}$ and $\lambda_{1}$ are {\em{fixed}} (see Figure \ref{figure_1}). Hence, in region III, $L(\beta, B^E)$ on the spectrum is determined by the limit when approaching the plane $\lambda_{1} + \lambda_{3} = \lambda_{2}$ along such lines from within region III. 

Using continuity of the spectrum in Hausdorff metric as well as joint continuity of the Lyapunov exponent in $(E,\lambda)$ (see Theorem \ref{thm_lana}), this limit can equally be determined by approaching the plane $\lambda_{1} + \lambda_{3} = \lambda_{2}$ from within region II. 

Based on a duality argument, in \cite{EE} we showed zero LE on the spectrum in region II for Diophantine $\beta$. For the readers convenience, we give a shortened, alternative proof of this fact in Appendix \ref{app_dualityconj}. 

In either case, the duality argument is based on \cite{E} which proves that the spectrum in the interior of region I is purely point with exponentially localized eigenfunctions. This however is established under the condition that $\beta$ is Diophantine. As far as the LE is concerned, we can remove the Diophantine condition using continuity in the frequency as asserted by Theorem \ref{thm_mainresult}.

\begin{theorem} \label{prop_regionIII}
Fix $\beta$ irrational. Then in region III, the Lyapunov exponent is zero on the spectrum.
\end{theorem}

Assembling the pieces, we finally obtain Theorem \ref{thm_mainresult}: The Theorems \ref{conj_duality} and \ref{prop_regionIII} establish zero Lyapunov exponent in region II and III. Let $\lambda \in I$. Invariance of the density of states under duality  (expressed in terms of the map $\sigma$, see (\ref{eq_inv}) in Appendix \ref{app_dualityconj}) and the Thouless formula for Jacobi operators imply \cite{D},
\begin{eqnarray} \label{eq_finalproof}
L(B_{\lambda}^{E},\beta) & = & - \int \log \abs{c_{\lambda}(x)} \ud x + \int \log\abs{E-E^{\prime}} \ud n_{\sigma(\lambda)}(E^{\prime}) \nonumber \\
 & = & \int \log \left \vert \frac{\lambda_{2} c_{\sigma(\lambda)}(x)}{c_{\lambda}(x)} \right \vert \ud x+ L(B_{\sigma(\lambda)}^{\lambda_2^{-1} E},\beta) ~\mbox{.}
\end{eqnarray}
Here, $c_{\lambda}$ denotes the function $c$ defined in (\ref{eq_hamiltonian1}), whereas in $c_{\sigma(\lambda)}$ we apply the map $\sigma$ to the coupling constants in $c$. Evaluating the integral in (\ref{eq_finalproof}) for $\lambda$ in region I (see Sec. 5 in \cite{D}) 
yields the formula for the Lyapunov exponent given in Theorem \ref{thm_mainresult}.

\section{Implications for general quasi-periodic, analytic one frequency Jacobi matrices} \label{sec_concl}

The example of extended Harper's equation not only illustrates a practicable method to obtain the Lyapunov exponent (or at least lower bounds) for analytic cocycles derived from a one-frequency almost periodic Jacobi matrix but also provides useful insights in how the theory originally developed  in \cite{B} for Schr\"odinger operators generalizes to the Jacobi case. 

In view of future applications to other Jacobi operators, in the present section we shall point out differences and peculiarities of the Jacobi case which relate to possible zeros of the determinant of the transfer matrix $A^E$ associated with the Jacobi operator. 

If not mentioned explicitly, in the following we consider a Jacobi operator on $l^2(\mathbb{Z})$ of the form (\ref{eq_hamiltonian}) with functions $c,v$ analytic on $\mathbb{T}_\delta$, $c \not \equiv 0$. 
%We shall refer to such an operator as {\em{almost periodic, analytic one frequency Jacobi matrix}}.

With a given Jacobi matrix one can associate the two transfer matrices $A^E$ and $B^E$ defined in (\ref{eq_deftransfer}). The former gives rise to an analytic cocyle whereas the latter is defined only a.e. on $\mathbb{T}_\delta$ due to possible zeros of $c(z)$. In view of the relation between the associated Lyapunov exponents, Eq. (\ref{eq_relAB}), we recall Lemma \ref{lem_quantaccbaby}.

\begin{enumerate}
\item It is crucial for the extrapolation process discussed in Sec.  \ref{sec_whatwecansay} to understand the analytic properties of  $L(\beta, A_\epsilon^E)$ w.r.t. $\epsilon$. In particular, by Proposition \ref{prop_convlyap} this function is convex in $\epsilon$. Convexity however fails in general for the Lyapunov exponent associated with the cocycle $(\beta, B_\epsilon^E)$. 

This can be seen on the example of Extended Harper's equation when $\lambda_1 = 0$. Then for $\lambda_3 = 1$ and $0 < \lambda_2 < 1$,
\begin{equation}
I_\epsilon(\lambda) = \begin{cases}  2 \pi \epsilon & \mbox{,} ~ \mathrm{e}^{2 \pi \epsilon} \geq \lambda_2 ~\mbox{,} \\
\log{\lambda_2} & \mbox{,} ~\mathrm{e}^{2 \pi \epsilon} \leq \lambda_2 ~\mbox{.} \end{cases}
\end{equation}
By  Lemma \ref{lem_quantaccbaby}  and the analytic properties of $L(\beta, A_\epsilon^E)$, $L(\beta, B_\epsilon^E)$ is still piecewise affine, however, since $L(\beta, A_\epsilon^E) = 2 \pi \abs{\epsilon}$ (Theorem \ref{prop_LE}), convexity fails at $\epsilon = \frac{1}{2 \pi} \log{\lambda_2}$: the right derivative changes from -1 to -2 (with increasing $\epsilon$). It can be checked that at this value of $\epsilon$ a zero of $c$ occurs. We mention that in general by  Lemma \ref{lem_quantaccbaby} convexity of $L(\beta, B_\epsilon^E)$ may only fail at those $\epsilon$ where $c$ exhibits a zero.

\item In Theorem \ref{thm_unifhyp} we assumed the analytic cocycle $(\beta, D)$ such that $\det(D)$ is bounded away from zero. That this assumption 
is indeed necessary may be seen when considering $(\beta, A_\epsilon^E)$ for Extended Harper's equation. To this end recall that $L(\beta, B_\epsilon^E) = L(\beta, (A_\epsilon^E)^\prime)$. For $L(\beta, B_\epsilon^E) > 0$, regularity of $(\beta, (A_\epsilon^E)^\prime)$ does in general not imply uniform hyperbolicity if $c(z)$ has real zeros: 

Consider $\lambda_1 = \lambda_3$ and $1 > 2 \lambda_1 >
\lambda_2$. Then $I_\epsilon(\lambda) = \log \lambda_3 + 2
\pi \abs{\epsilon}$ in a neighborhood of $\epsilon = 0$. Thus, by
(\ref{eq_defm}) and Theorem \ref{prop_LE}, $\omega (\beta,  B^E; \epsilon) = 0$ locally about $\epsilon = 0$  even though $L(\beta, B_\epsilon^E) > 0$ in the interior of region $I$.

\item Analysis of the Lyapunov exponent serves as an important ingredient in spectral analysis. For almost periodic Schr\"odinger operators with analytic potentials the almost reducibility conjecture (ARC) \cite{B,V,W,X} claims that sub-critical behavior of the Schr\"odinger cocycle on the spectrum implies purely absolutely continuous spectrum.

Section \ref{sec_reglyap} and above remarks suggest to formulate the ARC for Jacobi cocycles in terms of the analytic cocycle $(\beta, A_\epsilon^E)$ (and not for $(\beta, B_\epsilon^E)$). However, even for the analytic cocycle $(\beta, A_\epsilon^E)$ the ARC is false in general. 

For Extended Harper's equation with $\lambda_1 = 0$ and $ \lambda_3 = \lambda_2 \geq 1 $, $L(\beta, A_\epsilon^E)$ is regular (Corollary \ref{coro_regularity}) even though the spectrum has no ac component. The latter follows from Observation \ref{obs_cfun} and the following simple Proposition. 

We have learnt only recently that this result originally goes back to \cite{FF} \footnote{We are grateful to Barry Simon for this reference.}. Since the argument is short we include a sketch below. For details we refer to \cite{FF}.

\begin{prop}
Let $\lambda$ such that $c(z)$ has a real root. Then the ac spectrum ($\Sigma_{ac}$) is empty.
\end{prop}

\begin{proof}
We reduce to the following toy problem: Let 
\begin{equation*} 
(H \psi)_n := v_n \psi_{n} + c_n \psi_{n+1} + \overline{c_{n-1}} \psi_{n-1} ~\mbox{,} 
\end{equation*}
a Jacobi operator on $l^2(\mathbb{Z})$ and suppose $(c_n)_{n \in \mathbb{N}}$ is zero infinitely often. Then, the Jacobi matrix decouples into finite dimensional blocks whence the
spectrum is only pp. 

Suppose now $c(z)$ has a real zero, say at $x_0 \in \mathbb{T}$ and fix an irrational $\beta$. By ergodicity, for a.e. phase $\theta$ the rotational
trajectory comes arbitrarily close to $x_0$. Thus for such $\theta$ the operator  $H_{\lambda,\theta;\beta}$ can be considered a trace class perturbation of  above block Jacobi matrix. Invariance of the ac spectrum under trace class perturbations yields the claim.
\end{proof}

Using Observation \ref{obs_cfun}, we obtain a glimpse on the spectral theory for Extended Harper's equation:
\begin{coro}
Let $\beta$ be irrational. Then,
\begin{itemize}
\item[(i)] For $\lambda_3 \neq \lambda_1$, $\Sigma_{ac} = \emptyset$ whenever $\lambda_1 + \lambda_3 = \lambda_2$.
\item[(ii)] For $\lambda_3 = \lambda_1$, $\Sigma_{ac} = \emptyset$ along $2 \lambda_3 \geq\lambda_2$.
\end{itemize}
\end{coro}

Above example shows that, in general, ARC fails if $c(z)$ has zeros. We do however believe that the conjecture is still true whenever $\det A_\epsilon^E(x)$ is bounded away from zero: 

\begin{conj}[ARC for Jacobi cocycles]
For an analytic almost periodic Jacobi cocycles with determinant bounded away from zero, sub-critical behavior of $(\beta, A_\epsilon^E)$ implies (purely) absolutely continuous spectrum.
\end{conj}

\end{enumerate}

\appendix

\section{} \label{app_1}

\subsection{Convexity of $L(\beta, D_{\epsilon})$ w.r.t. $\epsilon$}
\begin{prop} \label{prop_convlyap}
Let $(\beta,D)$ be an analytic cocycle in the sense of Definition \ref{def_anacoc}. Then, $L(\beta, D_{\epsilon})$ is convex in $\epsilon$.
\end{prop}
\begin{proof}
The main point here is to convince the reader of convexity for {\em{singular}} analytic cocycles. For $D \in \mathcal{C}^\omega(\mathbb{T}, SL(2,\mathbb{C})$ (or more generally, for $D \in \mathcal{B}^\omega(\mathbb{T},M_2(\mathbb{C}))$), it is immediate that $L(\beta, D_{\epsilon})$ is subharmonic w.r.t. $\epsilon$ (by e.g. a Craig-Simon type argument \cite{CS}), in which case convexity easily follows, as pointed out in \cite{B}. We mention that the argument presented here can also be generalized to analytic cocycles on $\mathbb{T}^d$, $d \geq 1$.

Since convexity is preserved under pointwise limits, by the definition of the LE, (\ref{eq_defle_limit}), it suffices to verify that 
\begin{equation}
\frac{1}{n} \int_\mathbb{T} \log \norm{D^{(n)}(x+i\epsilon)} \ud x ~\mbox{,}
\end{equation}
is a convex function in $\epsilon$ for each $n \in \mathbb{N}$.

Considering $M_2(\mathbb{C})$ as equipped with Hilbert-Schmidt norm, this readily follows from:
\begin{lemma} \label{lem_subh}
Let $\{f_j\}_{j=1}^{N}$ be a finite sequence of functions holomorphic in some neighborhood of $\mathbb{S}^1$, and set $h_\infty:=\log \sum_{j=1}^{N} \abs{f_j}^2$.
Then, $J(\epsilon):=\int_\mathbb{T} h_\infty(\mathrm{e}^{2 \pi i x} \mathrm{e}^{-2\pi \epsilon}) \ud x$ is convex in $\epsilon$.
\end{lemma}
\begin{proof}
Follows by a simple approximation argument: For $\eta>0$, let \newline $h_\eta:=\log \left( \sum_{j=1}^{N} \abs{f_j}^2 + \eta \right)$ and correspondingly, $J_\eta(\epsilon):= \int_\mathbb{T} h_\eta(\mathrm{e}^{2 \pi i x} \mathrm{e}^{-2\pi i \epsilon}) \ud x$. Clearly, $J_\eta \searrow J_\infty$, as $\eta \to 0^+$, whence the claim is reduced to $J_\eta$. 

Direct computation shows that $h_\eta$ is a smooth, subharmonic function in a neighborhood of $\mathbb{S}^1$. 
\footnote{
If $\{g_j\}$ is a finite sequence of holomorphic functions with $\sum_j \abs{g_j(z)}^2 \neq 0$, one computes:
\begin{eqnarray}
\dfrac{\partial^2}{\partial \overline{z} \partial z} \log \sum_{j} \abs{g_j(z)}^2 & = & \frac{1}{\sum_{j=1}^{N} \abs{g_j}^2} \left\{\sum_{i \neq k} \abs{g_i^\prime}^2 \abs{g_k}^2 - 2 
\re \left[ \sum_{i < k} g_i^\prime \overline{g_i} g_k \overline{g_k^\prime} \right]\right\} \\
& \geq & \frac{1}{\sum_{j=1}^{N} \abs{g_j}^2} \sum_{i < k} \left( \abs{g_i^\prime} \abs{g_k} - \abs{g_k^\prime} \abs{g_i} \right)^2 \geq 0~\mbox{.}
\end{eqnarray}
}

Finally, considering $\epsilon$ as a {\em{complex}} variable, the identity
\begin{equation}
\dfrac{\partial^2}{\partial \overline{\epsilon} \partial \epsilon} = 4 \pi^2 \mathrm{e}^{-4\pi \re \epsilon}  \dfrac{\partial^2}{\partial \overline{z} \partial z} ~\mbox{,}
\end{equation}
implies that $J_\eta$ is a {\em{smooth}}, subharmonic function in $\epsilon$ which is constant in $\im \epsilon$, i.e. convex in $\re \epsilon$.
\end{proof}
\begin{remark}
Since subharmonicity is preserved under monotone decreasing limits, the proof of Lemma \ref{lem_subh} also shows that $J_\eta(\epsilon)$ and hence $L(\beta, D_{\epsilon})$ is subharmonic in $\epsilon$.
\end{remark}

\end{proof}

\section{Proof of Lemma \ref{lem_root}} \label{app_lemmaproof}
\begin{description}
\item[Step1] We first establish the statement for trigonometric polynomials, where $g$ can be computed explicitly. Since,
\begin{equation}
\sum_{j=-n}^{n} c_j \mathrm{e}^{2 \pi i j x} = \mathrm{e}^{- 2 \pi i n x} \sum_{j=0}^{2 n} c_{j-n} \mathrm{e}^{2 \pi i j x} ~\mbox{,}
\end{equation}
it suffices to consider trigonometric polynomials which are monic algebraic polynomials in the variable $\mathrm{e}^{2 \pi i x}$ of degree $r \geq1$. In particular, $f$ can be factorized according to its roots (counting multiplicity),
\begin{equation}
f =  \prod_{j=1}^{r} \left( \mathrm{e}^{2 \pi i x} - \mathrm{e}^{2 \pi i (x_j + i \epsilon_j)}   \right) ~\mbox{,}
\end{equation}
where for some $\delta>0$, $x_j \in [0,1)$ and $\abs{\epsilon_j} > \delta>0$, $1 \leq j \leq r$. 

Since $f$ is holomorphic and bounded from zero on the strip $\abs{\im z} \leq \delta$, there exists a holomorphic function $q$ defined on $\abs{\im z} < \delta$ such that 
\begin{equation} \label{eq_root3}
\mathrm{e}^q = f ~\mbox{;}
\end{equation}
in particular, $q$ satisfies
\begin{equation}
q^\prime(z) = \dfrac{f^\prime(z)}{f(z)} ~\mbox{,} ~ \abs{\im(z)} < \delta ~\mbox{.}
\end{equation}
On the other hand, for $x \in \mathbb{R}$,
\begin{equation} \label{eq_root1}
\dfrac{f^\prime(x)}{f(x)} = \sum_{j=1}^{r} \dfrac{ 2 \pi i \mathrm{e}^{2 \pi i x} }{\mathrm{e}^{2 \pi ix} - \mathrm{e}^{2 \pi i (x_j + i \epsilon_j)}} ~\mbox{.}
\end{equation}

Considering the individual summands in (\ref{eq_root1}), we compute for $1 \leq j \leq r$ and $x \in \mathbb{R}$
\begin{equation} \label{eq_root2}
\dfrac{2 \pi i \mathrm{e}^{2 \pi ix}}{\mathrm{e}^{2 \pi ix} - \mathrm{e}^{2 \pi i (x_j + i \epsilon_j)}} = \begin{cases} 
2 \pi i \sum_{n=0}^{\infty} \mathrm{e}^{2 \pi i x_j n} \mathrm{e}^{- 2 \pi n \epsilon_j} \mathrm{e}^{- 2 \pi i n x} & \mbox{if} ~ \epsilon_j>0 ~\mbox{,} \\
- 2 \pi i \sum_{n=1}^{\infty} \mathrm{e}^{- 2 \pi i n x_j} \mathrm{e}^{2\pi \epsilon_j n} \mathrm{e}^{2 \pi i nx} & \mbox{if} ~\epsilon_j < 0 ~\mbox{.}  \end{cases}
\end{equation}
%We mention that from basic complex function theory, the constant terms in the Laurent series in (\ref{eq_root2}) represent the number of zeros of $\left(z - \mathrm{e}^{2 \pi i (x_j + \epsilon_j)}\right)$ on the unit disc (modulo a factor of $2 \pi i$).

Hence, for $x\in \mathbb{R}$, $q$ is given by,
\begin{equation} \label{eq_root4}
q(x) - q(0) = \int_{0}^{x} \dfrac{f^\prime(t)}{f(t)} \ud t = 2 \pi i n_+ x + h(x) ~\mbox{,}
\end{equation}
where $n_+$ represents the number of $j \in \{1, \dots, n\}$ with $\epsilon_j >0$, and $h \in \mathcal{C}_\delta^\omega(\mathbb{R}/\mathbb{Z}; \mathbb{C})$. Notice by construction, $q(0) = \log f(0)$, where the branch of the $\log$ is fixed by (\ref{eq_root3}). 

Let now $g(z): = \mathrm{e}^{\frac{1}{2} q(z)}$, which by construction is holomorphic on $\abs{\im(z)} < \delta$. For $x \in \mathbb{R}$, the explicit expression in (\ref{eq_root4}) already implies that  $g(x)$ is 2-periodic satisfying $g^2(x) = f(x)$. Upon use of the uniqueness theorem for holomorphic functions, both these properties extend to all of $\abs{\im(z)} < \delta$.
\item[Step 2] If $f \in \mathcal{C}_\delta^\omega(\mathbb{R}/\mathbb{Z}; \mathbb{C})$ is not a trigonometric polynomial, let $\delta>0$ such that $f(z) \neq 0$, $\forall z \in \mathbb{T}_\delta$. Approximate $f$ uniformly on $\mathbb{T}_\delta$ by trigonometric polynomials $p_n$. 

From step 1, we obtain holomorphic $q_n$ on $\abs{\im(z)} < \delta$ such that
\begin{eqnarray}
g_n:=\mathrm{e}^{\frac{1}{2} q_n} \in \mathcal{C}_\delta^\omega(\mathbb{R}/2\mathbb{Z}; \mathbb{C}) ~\mbox{,} ~g_n^2 = p_n ~\mbox{and} \\ 
q_n(0) = \log p_n(0) \to \log f(0) ~\mbox{,} \label{eq_root5}
\end{eqnarray}
defined (eventually) with respect to a common branch of the log since $f(0) \neq 0$ and $p_n \to f$. 

Since $q_n$ satisfies the differential equation $q_n^\prime = \frac{p_n^\prime}{p_n}$ on $\abs{\im(z)} < \delta$, and by construction $p_n \to f$ uniformly with $f$ bounded from zero on $\mathbb{T}_\delta$, $(q_n^\prime)$ is easily seen to be uniformly Cauchy. Using (\ref{eq_root5}), this in turn implies uniform convergence of $(q_n)$ on $\mathbb{T}_\delta$.

Finally, letting $q:=\lim_{n \to \infty} q_n$ and defining $g:=\mathrm{e}^{\frac{1}{2} q}$ we obtain the lemma's claim.
\end{description}

\section{Proof of Lemma \ref{lemma_regunif}} \label{app_3}

The proof for the cocycle $(\beta,(A^{E})^{\prime})$ will at the same time prove the statement for $(\beta, B^{E})$.
Again, the function $\log{\abs{c(z)}}$ is harmonic on the strip $\abs{\im{z}} < \delta$, some $\delta > 0$. Harmonicity and Birkhoff's ergodic theorem imply that
$\forall x \in \mathbb{T}$
\begin{equation} \label{eq_lem1}
\dfrac{1}{n} \log \left \vert \prod_{j=0}^{n-1} \dfrac{c(x+ j \beta)}{\abs{\det A^{E}(x+j \beta)}^{1/2}} \right \vert \to 0 ~\mbox{,}
\end{equation}
as $n \to \infty$.

Suppose $(\beta, (A^{E})^{\prime})$ is uniformly hyperbolic for $E \in
\Sigma$. According to Definition \ref{def_unifhyp}, let
$u_{+}(x)$/$s_{+}(x)$ be the unstable/stable directions for the positive half-line solutions of
$H_{x} \psi = E \psi$, $x \in \mathbb{T}$. Note that uniform hyperbolicity on the positive half line implies uniform hyperbolicity on the negative half-line with $u_{+}=s_{-}$ and $u_{-}=s_{+}$. 

Let $w:=(\psi(0),\psi(-1)) \in \mathbb{C}^{2}, ~\norm{w} = 1$, be an arbitrary initial condition. For $x \in \mathbb{T}$ arbitrary, choose unit vectors $u_{x} \in u_{+}(x)$ and 
$s_{x} \in u_{-}(x)$. Then, $w=\beta_{x} u_{x} + \alpha_{x} s_{x}$ for some complex $\alpha_{x}, \beta_{x}$. Without loss of generality, we may assume $\beta_{x} \neq 0$ (otherwise apply the argument for iteration in the negative direction, since then $v=\alpha_{x} s_{x} \in u_{-}(x)$). 

Using (\ref{eq_lem1}),
\begin{equation*}
\frac{1}{n} \log \norm{ \prod_{j=n-1}^{0} (A^{E})^{'}(x+ j \beta)} = \frac{1}{n} \log \norm{\prod_{j=n-1}^{0} B^{E}(x+ j \beta)} + o(1) ~\mbox{, as $n \to \infty$.}
\end{equation*} 

Therefore, uniform hyperbolicity of $(\beta, (A^{E})^{\prime})$ implies that the solution of $H_{x} \psi = E \psi$ associated with the initial condition $w$ increases
exponentially in at least one direction. $w$ and $x$ were chosen arbitrarily, hence $\forall x \in \mathbb{T}$, $E$ is not a generalized eigenvalue of $H_{x}$. 
Since uniform hyperbolicity is an open condition, using above argument again, $\forall x \in \mathbb{T}$, E cannot be a limit point of generalized eigenvalues
of $H_{x}$ either. Hence, we obtain $E$ is not in the spectrum, a contradiction.

\section{Upper-semicontinuity of the LE; Proof of Theorem \ref{thm_uscle}} \label{app_uscle}
Since all matrix norms are equivalent, following it is convenient to equip $M_2(\mathbb{C})$ with the Hilbert-Schmidt norm. Note that Remark \ref{rem_counter} also applies to $\log_{+}$ and continuous $M_2(\mathbb{C})$-cocycles, i.e. for every $n \in \mathbb{N}_0$, 
$(\beta, D) \mapsto \int_{\mathbb{T}} \log_{+} \norm{D^{(n)}(x)} \ud x$ is continuous in $\mathcal{C}^0(\mathbb{T},M_2(\mathbb{C}))$.

On the other hand, by Fatou's lemma, $(\beta, D) \mapsto \int_{\mathbb{T}} \log_{-} \norm{D^{(n)}(x)} \ud x$ is upper-semicontinuous. Recalling (\ref{eq_defle_limit}), the statement of the theorem follows since the infimum over upper-semicontinuous functions is upper-semicontinuous.

\section{Proof of Lemma \ref{fact_conv}} \label{app_convlemma}
For $K \subset (0,1)$ compact, let $m:=\min K$ and $M:=\max K$. Fix $f \in \mathcal{F}$. By basic properties of convex functions, 
\begin{equation}
D_+(g)(M) \leq \dfrac{g(M+h) - g(M)}{h} ~\mbox{,}
\end{equation}
for any $g \in \mathcal{F}$ and $0<h<1-M$. 

In particular,
\begin{equation}
\inf_{\delta>0} \sup_{\norm{g-f}_{[0,1]}<\delta} D_+(g)(M) \leq \dfrac{f(M+h) - f(M)}{h} ~\mbox{,}
\end{equation}
for any $0<h<1-M$. Thus, taking $h\to0+$ we obtain:
\begin{equation}
\inf_{\delta>0} \sup_{\norm{g-f}_{[0,1]}<\delta} D_+(g)(M) \leq D_+(f)(M) ~\mbox{.}
\end{equation}

A similar argument shows,
\begin{equation}
D_-(f)(m) \leq \sup_{\delta>0} \inf_{\norm{g-f}_{[0,1]}<\delta} D_-(g)(m) ~\mbox{.}
\end{equation}

Recalling that $D_{\pm}(g)$ is monotone increasing on $[m,M]$, we obtain the claim.

\section{Duality} \label{app_dualityconj}

In view of the alternative proof of Theorem \ref{thm_mainresult}, we establish the missing link in Theorem \ref{prop_regionIII}:
\begin{theorem} \label{conj_duality}
Throughout region II, $L(\beta, B^{E})=0$ for all $E$ in the spectrum and irrational $\beta$.
\end{theorem}

Since the Hamiltonian in (\ref{eq_hamiltonian}) generalizes the almost Mathieu operator, it should not come as a surprise that its spectrum
expressed in terms of the density of states exhibits similar symmetry w.r.t. certain permutations of the coupling constants as the almost Mathieu operator. In fact, 
this so called duality has been explored for the present model in previous works \cite{D,E,EE,G,H}. A more general discussion of duality can be found in \cite{D}.

In this Section we will prove Theorem \ref{conj_duality}, adapting an idea of Delyon \cite{M}, that originally was used to establish
absence of point spectrum in the sub-critical region of the almost Mathieu equation. 
 
Let $\sigma$ denote the map, $\sigma(\lambda):=\frac{1}{\lambda_{2}}(\lambda_{3},1,\lambda_{1})$, where,
as earlier, $\lambda=(\lambda_{1},\lambda_{2},\lambda_{3})$. 

Referring to the partitioning of the parameter space introduced in Sec. \ref{sec_intro},  the map $\sigma$ acts as follows
\begin{equation} \label{eq_dualityofregions}
\sigma(I) = II~, ~\sigma(III) = III ~\mbox{.}
\end{equation}
In this sense, regions I and II are dual, whereas region III is self-dual. 

In \cite{E} (see Theorem 1 therein) it was proven that for Diophantine $\beta$ and $\lambda$ in region I, the spectrum of $H_{\lambda,\theta}$ 
is only pure point (i.e. for a certain set  $\mathbb{T}_{0}$ of phases $\theta$ with $\vert \mathbb{T}_{0} \vert =1$).
Fix $\lambda$ in region I and $\theta \in \mathbb{T}_{0}$. Let $\phi_{\theta}$ be a normalized eigenvector of $H_{\lambda,\theta}$ 
with associated eigenvalue $E_{\theta}$. If we let 
\begin{equation}
\hat{\phi_{\theta}}(x) = \sum_{n \in \mathbb{Z}} \phi_{\theta}(n) \mathrm{e}^{2 \pi i n x} ~\mbox{,}
\end{equation}
$\left\{ \mathrm{e}^{2 \pi i \theta m} \hat{\phi_{\theta}}(\eta + \beta m)   \right\}_{m \in \mathbb{Z}}$ defines a random sequence indexed
by points in $\mathbb{T}$. Set
\begin{equation}
\psi_{\theta}(\eta,m):= \mathrm{e}^{2 \pi i \theta m} \hat{\phi_{\theta}}(\eta + \beta m) ~\mbox{.}
\end{equation}

Since for each fixed $m \in \mathbb{Z}$, $\psi_{\theta}(.,m) \in \mathrm{L}^{2}(\mathbb{T})$, direct computation shows that
\begin{eqnarray} \label{eq_duality1}
& c_{\sigma{(\lambda)}}(. + \beta (m+1)) \psi_{\theta}(. , m+1) + \overline{c_{\sigma{(\lambda)}}(. + \beta (m-1))} \psi_{\theta}(. , m-1) & \nonumber \\
& + v(. + \beta m) \psi_{\theta}(. ,m)= \lambda_{2}^{-1} E_{\theta} \psi_{\theta}(. , m) ~\mbox{,} &
\end{eqnarray}
for each $m \in \mathbb{Z}$.

Equation (\ref{eq_duality1}) holds in the sense of $\mathrm{L}^{2}$-functions, whence it is true on a set $\Omega_{m}$ of full measure in $\mathbb{T}$. 
In particular, letting $\Omega:=\cap_{m \in \mathbb{Z}} \Omega_{m}$,  for each $\eta \in \Omega$ the sequence $\left\{ \psi_{\theta}(\eta,m) \right\}_{m \in \mathbb{Z}}$
satisfies the finite difference equation
\begin{equation}
H_{\sigma(\lambda),\eta} \psi_{\theta}(\eta, .) = \lambda_{2}^{-1} E_{\theta} \psi_{\theta}(\eta, .) ~\mbox{.}
\end{equation}

At the same time, for $\epsilon >0$ normalization of $\phi_{\theta}$ yields
\begin{eqnarray}
\int \ud \eta \sum_{n \in \mathbb{Z}} \left \vert n^{-\frac{1}{2}(1+\epsilon)} \psi_{\theta}(\eta,n) \right \vert^{2} = \sum_{n \in \mathbb{Z}} \dfrac{1}{n^{1+\epsilon}} < \infty ~\mbox{.}
\end{eqnarray}

Thus for every $\epsilon > 0$ and Lebesgue a.e. $\eta$ there exists $C_{\epsilon}(\eta) > 0$ such that 
\begin{equation} \label{eqn_subexpbound}
\abs{\psi_{\theta}(\eta,n)} \leq C_{\epsilon}(\eta)  n^{1/2(1+\epsilon)}   ~\mbox{,} ~\forall n \in \mathbb{Z} ~\mbox{.}
\end{equation}

It is a general fact for a family of bounded ergodic Jacobi operators $\left\{ H_{\omega} \right\}_{\omega \in \Omega}$, $(\Omega, \mathcal{F},\mathbb{P})$ a probability space, 
that fixing an energy $E$ \cite{N}, 
\begin{equation} \label{eq_ergodicjacobieigenv}
\mathbb{P}\{ E ~\mbox{is eigenvalue for} ~H_{\omega}\}=0 ~\mbox{.}
\end{equation}
Making use of a version of Osceledets Theorem for Jacobi operators, Theorem \ref{app_oseledec}, (\ref{eq_ergodicjacobieigenv}) implies
that for $\mathbb{P}$-a.e. $\omega$ any $\zeta \in \mathbb{C}^{\mathbb{Z}}$ satisfying $H_{\omega} \zeta = E \zeta$ has to exponentially increase in one direction
with rate given by the Lyapunov exponent.

The sub-exponential bound in (\ref{eqn_subexpbound}) then implies $L\left(B_{\sigma(\lambda)}^{\lambda_{2}^{-1} E_{\theta}},\beta\right) = 0$. 
Hence, if we let 
\begin{equation*}
\mathfrak{A} := \bigcup_{\theta \in \mathbb{T}_{0}} \sigma_{\scriptsize{pt}}(H_{\lambda;\theta}) ~\mbox{,}
\end{equation*}
$L\left(B_{\sigma(\lambda)}^{\lambda_{2}^{-1} E},\beta\right) = 0$ for all $E \in \mathfrak{A}$. In fact, continuity of the Lyapunov exponent w.r.t. the energy yields 
\begin{equation} \label{eq_duality2}
L\left(B_{\sigma(\lambda)}^{\lambda_{2}^{-1}E},\beta\right) = 0  ~, ~\forall E \in \overline{\mathfrak{A}} ~\mbox{.}
\end{equation}
Set $\mathfrak{B}:=\mathbb{R} \setminus \overline{\mathfrak{A}}$. 

Denote by $\ud n\left(\left\{H_{\lambda,\theta}\right\}; E\right)$ the density of states for the family of ergodic operators $\left\{H_{\lambda,\theta}\right\}$.
As shown in a more general context in \cite{D}, duality preserves the density of states. For the present model this yields,
\begin{equation} \label{eq_inv}
\ud n\left(\left\{H_{\lambda,\theta}\right\}; E\right) = \ud n\left( \left\{ \lambda_{2}  H_{\sigma(\lambda),\theta}\right\}; E\right) ~\mbox{.}
\end{equation}

Note that the spectrum of
$\left\{H_{\lambda,\theta}\right\}$ is contained in
$\overline{\mathfrak{A}}$, hence $ n\left(\left\{H_{\lambda,\theta}\right\}; \mathfrak{B}\right)=0$.
Using (\ref{eq_inv}), we obtain
\begin{equation} \label{eq_duality3}
\Sigma \left(\left\{\lambda_{2} H_{\sigma(\lambda)} \right\} \right) \subseteq \overline{\mathfrak{A}} ~\mbox{.}
\end{equation}

Finally, the spectral mapping theorem, (\ref{eq_duality2}), and (\ref{eq_duality3}) imply $L(B_{\sigma(\lambda)}^{E},\beta) = 0$ for any $E \in \Sigma \left(\left\{\lambda_{2} H_{\sigma(\lambda)} \right\} \right)$. Since $\lambda$
was chosen arbitrarily in region I, this implies Theorem \ref{conj_duality} for Diophantine $\beta$.

Finally, as mentioned in Sec. \ref{sec_alternative}, the Diophantine condition can be removed using continuity in the frequency as asserted by Theorem \ref{thm_lana}.

\section{Oseledets Theorem for Jacobi operators} \label{app_oseledec}
In Appendix \ref{app_dualityconj} we made use of Oseledets theorem to relate the Lyapunov exponent to the exponential decay rate of solution to the Schr\"odinger equation. The authors noticed that in the literature for Schr\"odinger operators, Oseledets theorem is usually given as a theorem for SL(2,$\mathbb{C}$) matrices even though Oseledets
original statement allows GL(2,$\mathbb{C}$) \cite{AA} .

The relation of one to the other is, of course, just a matter of appropriate re-normalization of the determinant. We include the simple argument below and show its consequences for ergodic Jacobi operators. A more general, deterministic statement which even allows for not invertible matrices can be found in \cite{GGGG}.

\begin{theorem}[deterministic formulation]
Let $(A_{n})_{n \in \mathbb{N}}$ be a sequence in GL(2,$\mathbb{C}$). Set $d_{n}:=(\det{A_{n}})^{\frac{1}{2}}, ~D_{n}:=\frac{1}{d_{n}}A_{n}, ~n \in \mathbb{N}$. 
Suppose that $m:= \lim_{n\to \infty} \frac{1}{n} \sum_{k=1}^{n} \log{\abs{d_{k}}}$ exists and is finite, $\lim_{n \to \infty} \frac{1}{n}\log{\norm{A_{n}}} = 0$, 
and that $L=L(A):=\lim_{n \to \infty} \frac{1}{n} \log{\norm{A_{n} \dots A_{1}}}$ exists and is finite with $L-m > 0$. 
Then, there is a one-dimensional subspace $S \subseteq \mathbb{C}^{2}$ such that for $v \in \mathbb{C}^{2}$
\begin{itemize}
\item[(i)] $\frac{1}{n} \log{\norm{A_{n} \dots A_{1} v}} \xrightarrow{n \to \infty} -L + 2 m$, if $v \in S$,
\item[(ii)] $\frac{1}{n} \log{\norm{A_{n} \dots A_{1} v}} \xrightarrow{n \to \infty} L $, if $v \in \mathbb{C}^{2} \setminus S$.
\end{itemize}
\end{theorem}
\begin{proof}
Since $m$ exists and is finite, $\frac{1}{n} \log{\abs{d_{n}}} \xrightarrow{n\to\infty} 0$, so $\lim_{n \to \infty} \frac{1}{n} \log{\norm{D_{n}}} = 0$. 
Note that $L-m=\lim_{n \to \infty} \frac{1}{n} \log{\norm{D_{n} \dots D_{1}}}$, hence by the SL(2,$\mathbb{C}$) version of Oseledec theorem, 
there exists a one-dimensional subspace $S \subseteq \mathbb{C}^{2}$ such that for $v \in \mathbb{C}^{2}$
\begin{eqnarray*}
\frac{1}{n} \log{\norm{A_{n} \dots A_{1} v}} = \frac{1}{n} \log{\norm{D_{n} \dots D_{1} v}} + \lim_{n\to \infty} \frac{1}{n} \sum_{k=1}^{n} \log{\abs{d_{k}}} \nonumber \\
\xrightarrow{n\to\infty} -L + 2m ~\mbox{, for} ~v \in S ~\mbox{,} \\
\frac{1}{n} \log{\norm{A_{n} \dots A_{1} v}} = \frac{1}{n} \log{\norm{D_{n} \dots D_{1} v}} + \lim_{n\to \infty} \frac{1}{n} \sum_{k=1}^{n} \log{\abs{d_{k}}} \nonumber \\
\xrightarrow{n\to\infty} L ~\mbox{, for} ~v \in \mathbb{C}^{2} \setminus S ~\mbox{.}
\end{eqnarray*}
\end{proof}
\begin{remark} \label{rem_osceledec}
\begin{itemize}
\item[(i)] Let $A: \Omega \to GL(2,\mathbb{C})$, where $(\Omega,\mathcal{F},\mathbb{P})$ is a probability space and $T:\Omega \to \Omega$ ergodic. If $A_{n}(x)=A(T^{n}x)$, $n \in \mathbb{N}$,
then by Birkhoff's ergodic theorem the Ceasaro mean $m$ becomes, $m=\mathbb{E}(\log\abs{d_{1}(x)})$. For an ergodic setup, we may hence replace the condition on finiteness
of $m$ by $\mathbb{E}(\log\abs{d_{1}(x)}) < \infty$. 
\item[(ii)] Consider a bounded ergodic Jacobi operator, i.e.
\begin{equation*} \label{eq_Jacobiop_1}
(H_{\omega} \psi)_n := b(T^{n} \omega) \psi_{n} + a(T^{n} \omega) \psi_{n+1} + \overline{a}(T^{n-1}\omega) \psi_{n-1} ~\mbox{,} 
\end{equation*}
where $T$ is ergodic on the probability space $(\Omega,\mathcal{F},\mathbb{P})$, $a: \Omega \to \mathbb{C}\setminus\{0\}, ~b:  \Omega \to \mathbb{R}$ are bounded random variables 
and $\mathbb{E}_{\omega}\log\abs{a(\omega)} < \infty$. 

In this case, above sequence $(A_{n})_{n \in \mathbb{Z}}$ is the associated sequence of one-step transfer matrices,
\begin{equation*} 
A_{n}(E)=\dfrac{1}{a(T^{n} \omega)}\begin{pmatrix} b(T^{n} \omega) - E & -\overline{a}(T^{n-1}\omega) \\
                 a(T^{n} \omega) & 0 \end{pmatrix}  ~\mbox{,}
\end{equation*}
where $E \in \mathbb{R}$. Then, since $T$ is measure preserving, (i) implies $m=0$.
\end{itemize}
\end{remark}

\bibliographystyle{amsplain}

\end{document}